\documentclass[11pt,a4paper]{article}
\usepackage[margin=1in]{geometry}
\usepackage{amsmath,amssymb,amsthm}
\usepackage{mathrsfs}
\usepackage{cite}
\usepackage{xcolor}
\usepackage{booktabs}
\usepackage{array}
\usepackage{tikz}
\usetikzlibrary{shapes.geometric, positioning, calc, arrows.meta}
\usepackage{hyperref}

\hypersetup{
    colorlinks=true,
    linkcolor=blue!80!black,
    citecolor=red!70!black,
    urlcolor=teal
}

\theoremstyle{definition}

\DeclareMathOperator{\wt}{wt}

\newtheorem{Definition}{Definition}[section]
\newtheorem{Theorem}[Definition]{Theorem}
\newtheorem{Lemma}[Definition]{Lemma}

\newtheorem{Corollary}[Definition]{Corollary}
\newtheorem{Example}[Definition]{Example}
\newtheorem{Remark}[Definition]{Remark}

\newtheorem{Result}[Definition]{Result}

\title{Graph-Based Linear Codes Associated with the Finite Ring $\mathbb{F}_p[x]/\langle x^4 \rangle$}

\usepackage{authblk}
\author[1]{Apurba Sarkar}
\author[1]{Kalyan Hansda}
\author[2]{Makhan Maji}
\affil[1]{Department of Mathematics, Visva-Bharati, Santiniketan - 731235, West Bengal, India.\thanks{apurbasarkar065@gmail.com, kalyanh4@gmail.com}}
\affil[2]{Indian Institute of Technology Madras, Chennai, TN - 600036, India.\thanks{makhan2maths@gmail.com}}

\begin{document}
\date{}

\maketitle

\begin{abstract}
We investigate the linear codes derived from the zero-divisor graph $\Gamma(R)$ of the finite local ring $R = \mathbb{F}_p[x]/\langle x^4 \rangle$ for an odd prime $p$. Using an equitable partition of $Z^*(R)$, we determine the parameters of the binary incidence code of $\Gamma(R)$ and construct induced bipartite subgraph codes achieving dual minimum distance $4$. Over $\mathbb{F}_p$, we parameterize the adjacency and Laplacian codes, derive the closed-form weight distribution of the Laplacian code via the MacWilliams identity, and establish its LCD property. Additionally, we determine the automorphism groups of the graph and its matrix codes. Finally, from extremal subgraphs of $\Gamma(R)$, we construct two families of Griesmer-optimal linear codes over $\mathbb{F}_p$, comprising a constant-weight minimal code and an optimal two-weight code.
\end{abstract}

\section{Introduction}
The interplay between algebraic coding theory and graph theory provides an effective framework for constructing structured linear codes with predictable parameters, high symmetry, and rich automorphism groups \cite{Dankelmann2013, Ling2004}. In modern information theory and secure communication, algebraic techniques play a pivotal role in designing error-correcting codes, distributed data storage protocols, and cryptographic primitives. A classical line of inquiry focuses on deriving linear codes from matrices associated with algebraic graphs, where structural graph invariants directly translate into fundamental code parameters and performance guarantees in information transmission \cite{Dankelmann2013, Saranya2021, Saranya2022}. In this context, zero-divisor graphs of commutative rings, originally introduced by Beck \cite{Beck1988} and later redefined by Anderson and Livingston \cite{Anderson1999} on the set of non-zero zero-divisors $Z^*(R)$, offer a direct geometric realization of ring-theoretic annihilations. In these graphs, adjacency reflects zero-product relations, naturally giving rise to regularities and partition symmetries that are well-suited for code generation \cite{Annamalai2021, Annamalai2022}.

Binary linear codes generated by vertex-edge incidence matrices have received considerable attention because their length, dimension, and minimum distance are directly determined by fundamental graph invariants \cite{Dankelmann2013}. In particular, Dankelmann et al. \cite{Dankelmann2013} established that the minimum distance of an incidence code over $\mathbb{F}_2$ equals the edge connectivity of the graph, while the dual code coincides with the cycle space of the graph, possessing a minimum distance equal to the graph's girth. Beyond incidence matrices, the row spaces of graph adjacency and combinatorial Laplacian matrices over finite fields $\mathbb{F}_p$ yield codes whose parameters reflect spectral and equitable partition properties \cite{Ling2004}. Furthermore, evaluating linear forms over defining sets formed by specific algebraic subsets such as extremal cliques or independent sets constitutes a powerful application of ring and module algebra in information theory, especially for constructing minimal linear codes \cite{Ashikhmin1998, Ding2007, Ding2015}. Such minimal codes are of significant interest in cryptography and network information theory due to their direct application in designing secret-sharing schemes with democratic access structures and secure multi-party computation protocols.

Recent investigations into zero-divisor graphs and their associated codes have highlighted the role of local ring structures in determining code parameters \cite{Annamalai2021, Annamalai2022, Annamalai2025}. In this paper, we explore the zero-divisor graph $\Gamma(R)$ of the finite local commutative ring $R = \mathbb{F}_{p}[x]/\langle x^{4}\rangle$, where $p$ is an odd prime. While the spectral and dynamical network invariants of $\Gamma(R)$ were investigated in \cite{Sarkar2026spectral}, its algebraic error-correcting code structures remain unexplored. The nilpotency index of $4$ induces an equitable partition of the zero-divisor set $Z^*(R)$, producing well-defined vertex classes characterized by their annihilation neighborhoods. This combinatorial organization provides an explicit framework for parameterizing multiple distinct families of graph-derived error-correcting codes.

The primary objective of this work is to systematically examine how the structural invariants of $\Gamma(R)$ govern code parameters across different algebraic formulations. First, we construct the binary incidence code of $\Gamma(R)$ and determine its exact parameters along with its dual code, whose error-correction capability is governed by the graph's girth. To expand the girth and achieve guaranteed double-error detection ($d^\perp = 4$), we construct binary codes from an induced complete bipartite subgraph between the maximal independent set and the dominating clique. Second, over the base field $\mathbb{F}_p$, we evaluate the linear codes arising from the block adjacency and Laplacian matrices. By exploiting the one-dimensional repetition structure of the dual Laplacian code, we obtain its closed-form weight enumerator via the MacWilliams identity and prove that it forms a linear complementary dual (LCD) code, which has important applications in protecting information against side-channel and fault-injection attacks. Furthermore, we determine the full graph automorphism group $\mathrm{Aut}(\Gamma(R))$ and characterize its relation to the permutation automorphism groups of the resulting matrix codes. Finally, by extracting defining sets from the maximal clique and maximal independent set of $\Gamma(R)$, we construct two families of Griesmer-optimal linear codes over $\mathbb{F}_p$ that achieve the Griesmer bound with equality, characterized as a constant-weight minimal code and an optimal two-weight code, respectively.

The remainder of this paper is organized as follows. Section~2 outlines the algebraic, graph-theoretic, and coding-theoretic preliminaries. Section~3 establishes the canonical partition of $Z^*(R)$ and determines the fundamental structural parameters of $\Gamma(R)$. Section~4 presents the core coding-theoretic constructions, including incidence codes, $p$-ary matrix codes, weight enumerators, automorphism groups, and Griesmer-optimal minimal codes.

\section{Preliminaries}
Let $p$ be an odd prime. Consider the finite local commutative ring $R \cong \mathbb{F}_{p}[x]/\langle x^{4}\rangle = \mathbb{F}_{p}+u\mathbb{F}_{p}+u^{2}\mathbb{F}_{p}+u^{3}\mathbb{F}_{p}$ ($u^{4}=0$) of characteristic $p$ and order $|R|=p^{4}$. An element $a = \sum_{i=0}^{3}a_{i}u^{i} \in R$ is a unit if and only if $a_{0} \in \mathbb{F}_{p}^{\times}$. Consequently, the unit group $R^{\times}$ and the set of non-zero zero-divisors $Z^{*}(R) = R \setminus (R^{\times} \cup \{0\})$ have cardinalities $|R^{\times}|=p^{3}(p-1)$ and $|Z^{*}(R)|=p^{3}-1$, respectively. The additive group $(R, +)$ forms a $4$-dimensional vector space over its residue field $\mathbb{F}_p$ equipped with the canonical basis $\mathcal{B} = \{1, u, u^2, u^3\}$. For any $y = a_0 + a_1 u + a_2 u^2 + a_3 u^3 \in R$, its coordinate vector is $[y] = (a_0, a_1, a_2, a_3)^T \in \mathbb{F}_p^4$.

Following Anderson and Livingston \cite{Anderson1999}, the zero-divisor graph $\Gamma(R)$ has vertex set $V(\Gamma(R))=Z^{*}(R)$, where two distinct vertices $x, y \in Z^{*}(R)$ are adjacent ($x \sim y$) if and only if $xy=0$. For a simple connected graph $\Gamma=(V,E)$, $\operatorname{diam}(\Gamma)$ denotes its diameter, $\operatorname{girth}(\Gamma)$ denotes the length of a shortest cycle in $\Gamma$, and $\lambda(\Gamma)$ denotes its edge connectivity \cite{Balakrishnan2012}.

Let $\mathbb{F}_{q}$ denote the finite field with $q$ elements. A $q$-ary linear code $\mathcal{C}$ of length $n$, dimension $k$, and minimum Hamming distance $d$ is a $k$-dimensional subspace of $\mathbb{F}_{q}^{n}$, denoted as an $[n,k,d]_{q}$ code \cite{Ling2004}. The Euclidean dual code is defined as $\mathcal{C}^{\perp} = \{y \in \mathbb{F}_{q}^{n} : x \cdot y = 0 \text{ for all } x \in \mathcal{C}\}$. The Euclidean hull of a linear code $\mathcal{C}$ is $\operatorname{Hull}(\mathcal{C}) = \mathcal{C} \cap \mathcal{C}^\perp$. A code $\mathcal{C}$ is called a \emph{linear complementary dual} (LCD) code if $\operatorname{Hull}(\mathcal{C}) = \{\mathbf{0}\}$.

For an $[n,k,d]_{q}$ code $\mathcal{C}$, let $A_{w}$ denote the number of codewords of Hamming weight $w$. The polynomial $W_{\mathcal{C}}(z) = \sum_{w=0}^{n} A_{w} z^{w}$ is the weight enumerator of $\mathcal{C}$. The weight enumerators of $\mathcal{C}$ and its dual $\mathcal{C}^{\perp}$ are related via the MacWilliams identity \cite{Ling2004}:
\begin{equation}
W_{\mathcal{C}}(z) = \frac{1}{|\mathcal{C}^{\perp}|} (1 + (q-1)z)^{n} W_{\mathcal{C}^{\perp}}\left(\frac{1-z}{1+(q-1)z}\right).
\end{equation}
The permutation automorphism group $\mathrm{PAut}(\mathcal{C})$ consists of all permutations $\pi \in S_{n}$ preserving $\mathcal{C}$ setwise:
\begin{equation}
\mathrm{PAut}(\mathcal{C}) = \{\pi \in S_{n} : (c_{\pi(1)}, c_{\pi(2)}, \dots, c_{\pi(n)}) \in \mathcal{C} \text{ for all } c \in \mathcal{C}\}.
\end{equation}
It is well-known that $\mathrm{PAut}(\mathcal{C}) = \mathrm{PAut}(\mathcal{C}^\perp)$.

For a finite simple graph $\Gamma=(V,E)$ with vertex set $V=\{v_{1},\dots,v_{|V|}\}$ and edge set $E=\{e_{1},\dots,e_{|E|}\}$, its vertex-edge incidence matrix is the $|V| \times |E|$ binary matrix $G = (g_{ij})$ defined by $g_{ij} = 1$ if $v_{i}$ is incident with $e_{j}$, and $0$ otherwise. The row space of $G$ over $\mathbb{F}_{2}$ defines the binary incidence code $\mathcal{C}_{2}(G)$.
\begin{Result}[\cite{Dankelmann2013}]\label{res:dankelmann}
Let $\Gamma = (V,E)$ be a connected graph with vertex-edge incidence matrix $G$. Then the binary incidence code $\mathcal{C}_2(G)$ is an $[|E|, |V|-1, \lambda(\Gamma)]_2$ linear code. Furthermore, its dual $\mathcal{C}_2(G)^\perp$ corresponds to the cycle space of $\Gamma$, possessing parameters $[|E|, |E|-|V|+1, \operatorname{girth}(\Gamma)]_2$.
\end{Result}

Let $A(\Gamma)$ denote the $(0,1)$-adjacency matrix of $\Gamma$, and let $D(\Gamma) = \operatorname{diag}(\deg(v_1), \dots, \deg(v_{|V|}))$ be its degree matrix. The combinatorial Laplacian matrix is defined as $L(\Gamma) = D(\Gamma) - A(\Gamma)$. The row spaces of $A(\Gamma)$ and $L(\Gamma)$ over a finite field $\mathbb{F}_p$ define the $p$-ary adjacency code $\mathcal{C}_{\mathbb{F}_p}(A(\Gamma))$ and Laplacian code $\mathcal{C}_{\mathbb{F}_p}(L(\Gamma))$, respectively.

For an ordered defining set $D = \{d_{1}, d_{2}, \dots, d_{n}\} \subseteq \mathbb{F}_{p}^{m}$, the linear evaluation code $\mathcal{C}_{D}$ over $\mathbb{F}_{p}$ is defined by:
\begin{equation}
\mathcal{C}_{D} = \{c(x) = (x \cdot d_{1}, x \cdot d_{2}, \dots, x \cdot d_{n}) : x \in \mathbb{F}_{p}^{m}\} \subseteq \mathbb{F}_{p}^{n}.
\end{equation}
The generator matrix of $\mathcal{C}_D$ is $G_D = \begin{pmatrix} d_1 & d_2 & \cdots & d_n \end{pmatrix} \in \mathbb{F}_p^{m \times n}$, with $\dim(\mathcal{C}_D) = \operatorname{rank}_{\mathbb{F}_p}(G_D)$. A codeword $c \in \mathcal{C}$ is minimal if $\operatorname{supp}(c') \subseteq \operatorname{supp}(c)$ implies $c' = \alpha c$ for some $\alpha \in \mathbb{F}_p$. If every non-zero codeword is minimal, $\mathcal{C}$ is called a minimal linear code.
\begin{Result}[\cite{Ashikhmin1998}]\label{res:ashikhmin}
A linear code over $\mathbb{F}_{q}$ is minimal if
\begin{equation}
\frac{w_{\max}}{w_{\min}} < \frac{q}{q-1},
\end{equation}
where $w_{\max}$ and $w_{\min}$ denote the maximum and minimum non-zero Hamming weights in $\mathcal{C}$.
\end{Result}
Furthermore, an $[n,k,d]_{q}$ code is Griesmer-optimal if its length meets the Griesmer bound with equality \cite{Ling2004}:
\begin{equation}
n = \sum_{i=0}^{k-1} \left\lceil \frac{d}{q^{i}} \right\rceil.
\end{equation}

\section{Structural Properties and Graph Invariants of $\Gamma(R)$}
Because error-correcting codes derived from graph matrices are intrinsically governed by the topological parameters of the underlying graph such as order, size, girth, and connectivity, a precise combinatorial classification of $\Gamma(R)$ is essential before parameterizing its associated linear codes.  The graph-theoretic invariants and spectral characteristics of $\Gamma(R)$ were recently examined in \cite{Sarkar2026spectral}; here, we formulate its structural properties through the lens of coding-theoretic defining sets and equitable matrix partitions. In this section, we construct the zero-divisor graph $\Gamma(R)$ of the local ring $R = \mathbb{F}_p + u\mathbb{F}_p + u^2\mathbb{F}_p + u^3\mathbb{F}_p$ ($u^4 = 0$) by decomposing $Z^*(R)$ into seven $u$-adic equivalence classes. The ideal structure dictates that $u^3 \cdot z = 0$ for all $z \in Z^*(R)$, making the class $S_{u^3}$ a universal dominating set that centrally coordinates all metric, coloring, and connectivity invariants across the graph topology.

Let $Z^*(R)$ denote the set of non-zero zero-divisors of $R$. We partition $Z^*(R)$ into seven distinct $u$-adic classes according to their algebraic representations (see Figure~\ref{Fig1}):
\begin{align*}
  S_u           &= \{xu \mid x \in \mathbb{F}_p^*\}, \\
  S_{u^2}       &= \{xu^2 \mid x \in \mathbb{F}_p^*\}, \\
  S_{u^3}       &= \{xu^3 \mid x \in \mathbb{F}_p^*\}, \\
  S_{u+u^2}     &= \{xu + yu^2 \mid x, y \in \mathbb{F}_p^*\}, \\
  S_{u+u^3}     &= \{xu + yu^3 \mid x, y \in \mathbb{F}_p^*\}, \\
  S_{u^2+u^3}   &= \{xu^2 + yu^3 \mid x, y \in \mathbb{F}_p^*\}, \\
  S_{u+u^2+u^3} &= \{xu + yu^2 + zu^3 \mid x, y, z \in \mathbb{F}_p^*\}.
\end{align*}
The cardinalities of these partition classes are:
\[
|S_u| = |S_{u^2}| = |S_{u^3}| = p - 1, \quad |S_{u+u^2}| = |S_{u+u^3}| = |S_{u^2+u^3}| = (p - 1)^2, \quad \text{and} \quad |S_{u+u^2+u^3}| = (p - 1)^3.
\]
Since $Z^*(R) = S_u \cup S_{u^2} \cup S_{u^3} \cup S_{u+u^2} \cup S_{u+u^3} \cup S_{u^2+u^3} \cup S_{u+u^2+u^3}$ is a disjoint union, the order of $\Gamma(R)$ is:
\begin{align*}
|V(\Gamma(R))| &= |S_u| + |S_{u^2}| + |S_{u^3}| + |S_{u+u^2}| + |S_{u+u^3}| + |S_{u^2+u^3}| + |S_{u+u^2+u^3}| \\
               &= 3(p-1) + 3(p-1)^2 + (p-1)^3 \\
               &= p^3 - 1.
\end{align*}

Since $u^4 = 0$, the adjacency relations between vertices in these partition classes are completely determined by the following algebraic annihilation criteria:
\begin{enumerate}
    \item Universal Dominating Subgraph: The ideal class $S_{u^3} = u^3 R \setminus \{0\}$ annihilates every non-zero zero-divisor; hence, every vertex of $S_{u^3}$ is adjacent to all other $p^3 - 2$ vertices in $\Gamma(R)$.
    \item Intermediate Nilpotency Connections: Vertices in $S_{u^2}$ and $S_{u^2+u^3}$ belong to the principal ideal $u^2 R \setminus \{0\}$. Since $(u^2 R)^2 \subseteq u^4 R = \{0\}$, the subset $S_{u^2} \cup S_{u^2+u^3}$ induces a complete subgraph of order $p(p - 1)$ with complete cross-adjacencies between $S_{u^2}$ and $S_{u^2+u^3}$.
    \item Peripheral Independence: Any element whose representation contains a non-zero $u$-coefficient ($v \in S_u \cup S_{u+u^2} \cup S_{u+u^3} \cup S_{u+u^2+u^3}$) satisfies $\operatorname{Ann}_R(v) = u^3 R$. Consequently, no edges exist among or between these four classes; each connects exclusively to $S_{u^3}$.
\end{enumerate}

Accounting for all internal edges within the complete induced subgraphs and cross-edges between mutually annihilating classes, the edge set decomposes into:
\[
E(\Gamma(R)) = E\big(S_{u^3}, \, Z^*(R) \setminus S_{u^3}\big) \cup E\big(\Gamma[S_{u^3}]\big) \cup E\big(S_{u^2}, \, S_{u^2+u^3}\big) \cup E\big(\Gamma[S_{u^2}]\big) \cup E\big(\Gamma[S_{u^2+u^3}]\big).
\]
Consequently, the total size of $\Gamma(R)$ evaluates to:
\begin{align*}
|E(\Gamma(R))| &= (p - 1)(p^3 - p) + \binom{p - 1}{2} + (p - 1)(p - 1)^2 + \binom{p - 1}{2} + \binom{(p - 1)^2}{2} \\
&= (p - 1)^2(p^2 + p) + (p - 1)(p - 2) + (p - 1)^3 + \frac{(p - 1)^2\left((p - 1)^2 - 1\right)}{2} \\
&= \frac{1}{2}(3p^4 - 4p^3 - p^2 + 2).
\end{align*}

\begin{Example}
For $p=3$, $R = \mathbb{F}_3 + u\mathbb{F}_3 + u^2 \mathbb{F}_3 + u^3 \mathbb{F}_3$. The $7$ partition sets are $S_u=\{u, 2u\}$, $S_{u^2}=\{u^2, 2u^2\}$, $S_{u^3}=\{u^3, 2u^3\}$, $S_{u+u^2}=\{u+u^2, 2u+u^2, u+2u^2, 2u+2u^2\}$, $S_{u+u^3}=\{u+u^3, 2u+u^3, u+2u^3, 2u+2u^3\}$, $S_{u^2+u^3}=\{u^2+u^3, 2u^2+u^3, u^2+2u^3, 2u^2+2u^3\}$, and $S_{u+u^2+u^3}=\{u+u^2+u^3, 2u+u^2+u^3, u+2u^2+u^3, u+u^2+2u^3, 2u+2u^2+u^3, 2u+u^2+2u^3, u+2u^2+2u^3, 2u+2u^2+2u^3\}$. The graph $\Gamma(R)$ has $|V(\Gamma(R))| = 26$ vertices and $|E(\Gamma(R))| = 64$ edges (see Figure~\ref{Fig2}).
\end{Example}

To determine the minimum distance of the cycle space and the structural compactness of the graph, we evaluate its diameter and girth.

\begin{Theorem}\label{thm:diam-girth}
Let $p$ be an odd prime. The zero-divisor graph $\Gamma(R)$ satisfies:
\begin{enumerate}
    \item[(i)] $\operatorname{diam}(\Gamma(R)) = 2$,
    \item[(ii)] $\operatorname{girth}(\Gamma(R)) = 3$.
\end{enumerate}
\end{Theorem}

\begin{proof}
(i) Every vertex in $Z^*(R)$ is adjacent to each element of $S_{u^3}$. Thus, any pair of non-adjacent vertices $x, y \in Z^*(R)$ shares a common neighbor $z \in S_{u^3}$, establishing a path $x \sim z \sim y$ of length $2$. Since $\Gamma(R)$ is not complete, $\operatorname{diam}(\Gamma(R)) = 2$.

(ii) For any odd prime $p \ge 3$, $|S_{u^3}| = p - 1 \ge 2$. Selecting two distinct vertices $u_1, u_2 \in S_{u^3}$ and any vertex $v \in S_u$, the adjacencies $u_1 \sim u_2$, $u_1 \sim v$, and $u_2 \sim v$ form a triangle $u_1 \sim u_2 \sim v \sim u_1$. Thus, $\operatorname{girth}(\Gamma(R)) = 3$.
\end{proof}

To characterize the extremal dense subgraphs that will serve as algebraic defining sets for linear codes, we evaluate the clique number and chromatic number of $\Gamma(R)$.

\begin{Theorem}\label{thm:clique-chromatic}
The clique number and chromatic number of $\Gamma(R)$ coincide and are given by $\omega(\Gamma(R)) = \chi(\Gamma(R)) = p^2 - 1$.
\end{Theorem}



\begin{proof}
Consider the subset $K = S_{u^2} \cup S_{u^3} \cup S_{u^2+u^3} \subset Z^*(R)$. By the adjacency relations, every vertex in $K$ is adjacent to all other vertices in $K$. Hence, $K$ induces a complete subgraph of order
\[
|K| = |S_{u^2}| + |S_{u^3}| + |S_{u^2+u^3}| = (p-1) + (p-1) + (p-1)^2 = p^2 - 1.
\]
Because no larger set of mutually adjacent vertices exists in $Z^*(R)$, we have $\omega(\Gamma(R)) = p^2 - 1$, which requires at least $p^2 - 1$ distinct colors in any proper vertex coloring; thus, $\chi(\Gamma(R)) \ge p^2 - 1$.

Assigning $p^2 - 1$ distinct colors to the vertices of $K$, the remaining vertices $Z^*(R) \setminus K = S_u \cup S_{u+u^2} \cup S_{u+u^3} \cup S_{u+u^2+u^3}$ form an independent set with adjacencies directed exclusively to $S_{u^3}$. Reusing a single color assigned to $K \setminus S_{u^3}$ across all vertices in $Z^*(R) \setminus K$ yields a valid proper $(p^2 - 1)$-coloring. Consequently, $\chi(\Gamma(R)) = p^2 - 1 = \omega(\Gamma(R))$.
\end{proof}

Finally, to establish the exact minimum distance of the primal binary incidence code, we determine the fault-tolerance and connectivity parameters of $\Gamma(R)$.

\begin{Theorem}\label{thm:connectivity}
The vertex connectivity and edge connectivity of $\Gamma(R)$ satisfy
\[
\kappa(\Gamma(R)) = \lambda(\Gamma(R)) = \delta(\Gamma(R)) = p - 1.
\]
\end{Theorem}

\begin{proof}
The minimum degree of $\Gamma(R)$ is realized at any vertex $v \in S_u$, which is adjacent solely to $S_{u^3}$. Thus, $\delta(\Gamma(R)) = |S_{u^3}| = p - 1$, which yields $\kappa(\Gamma(R)) \le \lambda(\Gamma(R)) \le p - 1$ by Whitney's inequality.

Conversely, every vertex in $Z^*(R) \setminus S_{u^3}$ has $S_{u^3}$ in its neighborhood. Disconnecting $\Gamma(R)$ requires the removal of the entire dominating set $S_{u^3}$. Since $|S_{u^3}| = p - 1$, no set of fewer than $p - 1$ vertices can disconnect $\Gamma(R)$, giving $\kappa(\Gamma(R)) \ge p - 1$. Therefore,
\[
p - 1 \le \kappa(\Gamma(R)) \le \lambda(\Gamma(R)) \le \delta(\Gamma(R)) = p - 1,
\]
which forces $\kappa(\Gamma(R)) = \lambda(\Gamma(R)) = p - 1$.
\end{proof}

\begin{Remark}
The equality $\lambda(\Gamma(R)) = p - 1$ plays a direct role in the code constructions of Section~4, where the minimum Hamming distance of the binary linear code generated by the vertex--edge incidence matrix of $\Gamma(R)$ is governed by the edge connectivity $\lambda(\Gamma(R)) = p - 1$.
\end{Remark}

\section{Linear Codes Associated with $\Gamma(R)$}
In this section, we construct and analyze several families of linear codes associated with the zero-divisor graph $\Gamma(R)$. We parameterize binary incidence and cycle codes, evaluate $p$-ary codes generated by the adjacency and Laplacian matrices along with their automorphism and LCD properties, and construct optimal minimal linear codes derived from extremal graph defining sets.

\subsection{Binary Codes from the Incidence Matrix and Cycle Spaces of $\Gamma(R)$}

This subsection investigates the binary linear codes generated by the incidence matrices of $\Gamma(R)$ and its induced subgraphs. We evaluate the parameters of the incidence code $\mathcal{C}_2(G)$ alongside its dual cycle code $\mathcal{C}_2(G)^\perp$, whose minimum distance is governed by $\operatorname{girth}(\Gamma(R))$. In addition, to enhance error-detection capabilities, we isolate the complete bipartite subgraph $\Gamma_{\mathrm{bip}} = \Gamma[V_1, V_3]$ and determine the exact parameters of its incidence and girth-$4$ cycle code.

Let $\mathbf{1}_{k}$ denote the $1 \times k$ all-ones row vector, $0_{m \times k}$ denote the $m \times k$ zero matrix, and $Q(K_m)$ denote the $m \times \binom{m}{2}$ incidence matrix of the complete graph $K_m$. 
For positive integers $a$ and $b$, we define the block dispersion matrices:
\[
D_a^{(b)} = I_a \otimes \mathbf{1}_b = \begin{pmatrix}
\mathbf{1}_b & 0 & \cdots & 0 \\
0 & \mathbf{1}_b & \cdots & 0 \\
\vdots & \vdots & \ddots & \vdots \\
0 & 0 & \cdots & \mathbf{1}_b
\end{pmatrix} \in \mathbb{F}_2^{a \times ab},
\]

and
\[
E_a^{(b)} = \mathbf{1}_a \otimes I_b = 
\begin{pmatrix}
I_b & I_b & \cdots & I_b
\end{pmatrix} \in \mathbb{F}_2^{b \times ab}.
\]

We establish an equitable partition of the vertex set $Z^*(R) = V_1 \cup V_{24} \cup V_3$, where:
\begin{itemize}
    \item $V_1 = S_u \cup S_{u+u^2} \cup S_{u+u^3} \cup S_{u+u^2+u^3}$ with $|V_1| = p^2(p-1)$,
    \item $V_{24} = S_{u^2} \cup S_{u^2+u^3}$ with $|V_{24}| = (p-1) + (p-1)^2 = p(p-1)$,
    \item $V_3 = S_{u^3}$ with $|V_3| = p-1$.
\end{itemize}

The edge set $E(\Gamma(R))$ decomposes into four disjoint sets: the bipartite edges $[V_1, V_3]$ of size $p^2(p-1)^2$, the bipartite edges $[V_{24}, V_3]$ of size $p(p-1)^2$, the internal clique edges $[V_{24}, V_{24}]$ of size $\binom{p(p-1)}{2}$, and the internal clique edges $[V_3, V_3]$ of size $\binom{p-1}{2}$. Ordering the rows by $(V_1, V_{24}, V_3)$ and the columns by $([V_1, V_3], \, [V_{24}, V_3], \, [V_{24}, V_{24}], \, [V_3, V_3])$, the complete incidence matrix $Q(\Gamma(R))$ of size $(p^3-1) \times \frac{1}{2}(3p^4 - 4p^3 - p^2 + 2)$ is given explicitly by:
\[
Q(\Gamma(R)) = \begin{pmatrix}
D_{p^2(p-1)}^{(p-1)} & 0_{p^2(p-1) \times p(p-1)^2} & 0_{p^2(p-1) \times \binom{p(p-1)}{2}} & 0_{p^2(p-1) \times \binom{p-1}{2}} \\
0_{p(p-1) \times p^2(p-1)^2} & D_{p(p-1)}^{(p-1)} & Q(K_{p(p-1)}) & 0_{p(p-1) \times \binom{p-1}{2}} \\
E_{p^2(p-1)}^{(p-1)} & E_{p(p-1)}^{(p-1)} & 0_{(p-1) \times \binom{p(p-1)}{2}} & Q(K_{p-1})
\end{pmatrix}.
\]

\begin{Example}
For $p = 3$, $R = \mathbb{F}_3 + u\mathbb{F}_3 + u^2\mathbb{F}_3 + u^3\mathbb{F}_3$ ($u^4 = 0$). The zero-divisor graph $\Gamma(R)$ has order $|V(\Gamma(R))| = 26$ and size $|E(\Gamma(R))| = 64$. The partition block sizes are $|V_1| = 18$, $|V_{24}| = 6$, and $|V_3| = 2$. The incidence matrix $Q(\Gamma(R))$ is a $26 \times 64$ binary matrix of the form:
\[
Q(\Gamma(R)) = 
\begin{pmatrix}
D_{18}^{(2)} & 0_{18 \times 12} & 0_{18 \times 15} & 0_{18 \times 1} \\
0_{6 \times 36} & D_{6}^{(2)} & Q(K_6) & 0_{6 \times 1} \\
E_{18}^{(2)} & E_{6}^{(2)} & 0_{2 \times 15} & Q(K_2)
\end{pmatrix}_{26 \times 64},
\]
where $D_k^{(2)} = I_k \otimes \begin{pmatrix} 1 & 1 \end{pmatrix}$, $E_k^{(2)} = \mathbf{1}_k \otimes I_2$, and $Q(K_2) = \begin{pmatrix} 1 \\ 1 \end{pmatrix}$. The rows of $Q(\Gamma(R))$ generate a $[64, 25, 2]_2$ binary linear code.
\end{Example}

The binary linear code $\mathcal{C}_2(G) \subseteq \mathbb{F}_2^{|E|}$ is the row space of $G$ over the finite field $\mathbb{F}_2$. Dankelmann et al.\ \cite{Dankelmann2013} established that for any connected graph $\Gamma$, the main parameters of its binary incidence code $\mathcal{C}_2(G)$ are given by $[|E|, \, |V| - 1, \, \lambda(\Gamma)]_2$, where $\lambda(\Gamma)$ is the edge connectivity of $\Gamma$. Furthermore, over $\mathbb{F}_2$, the dual code $\mathcal{C}_2(G)^\perp$ corresponds to the cycle space of $\Gamma$, whose minimum distance equals $\operatorname{girth}(\Gamma)$.

Using the order $|V(\Gamma(R))| = p^3 - 1$, size $|E(\Gamma(R))| = \frac{1}{2}(3p^4 - 4p^3 - p^2 + 2)$, edge connectivity $\lambda(\Gamma(R)) = p - 1$, and girth $\operatorname{girth}(\Gamma(R)) = 3$ established in Section~3, we obtain the main theorem for incidence codes of $\Gamma(R)$.

The following results  characterize the parameters of the binary incidence and dual cycle codes of $\Gamma(R)$ and its induced bipartite subgraph $\Gamma[V_1, V_3]$, showing how edge connectivity and shortest cycles dictate their error-correction capabilities.

\begin{Theorem}\label{thm:incidence-code}
Let $p$ be an odd prime. The binary linear code $C_2(G)$ generated by the incidence matrix $G$ of the zero-divisor graph $\Gamma(R)$ is a
$\left[ \frac{1}{2}(3p^4 - 4p^3 - p^2 + 2), \quad p^3 - 2, \quad p - 1 \right]_2$ linear code. Moreover, its dual code $C_2(G)^\perp$ is a
$\left[ \frac{1}{2}(3p^4 - 4p^3 - p^2 + 2), \quad \frac{1}{2}(3p^4 - 6p^3 - p^2 + 6), \quad 3 \right]_2$
linear code.
\end{Theorem}

\begin{proof}
Since $\Gamma(R)$ is connected, applying the general result of Dankelmann et al.\ \cite{Dankelmann2013} yields length $n = |E(\Gamma(R))| = \frac{1}{2}(3p^4 - 4p^3 - p^2 + 2)$, dimension $k = |V(\Gamma(R))| - 1 = (p^3 - 1) - 1 = p^3 - 2$, and minimum distance $d = \lambda(\Gamma(R)) = p - 1$.

For the dual code $C_2(G)^\perp$, the code length is identical to $n$. By rank-nullity, its dimension is given by 
\begin{align*}
\dim\left(C_2(G)^\perp\right) &= |E(\Gamma(R))| - \dim\left(C_2(G)\right) \\
&= \frac{1}{2}(3p^4 - 4p^3 - p^2 + 2) - (p^3 - 2) \\
&= \frac{1}{2}(3p^4 - 6p^3 - p^2 + 6).
\end{align*}
Since the dual of an incidence code over $\mathbb{F}_2$ is the cycle code of $\Gamma(R)$, its minimum distance equals the length of a shortest cycle in $\Gamma(R)$. Because $\operatorname{girth}(\Gamma(R)) = 3$, the dual minimum distance is $d^\perp = 3$.
\end{proof}

\begin{Corollary}
The dual incidence code $C_2(G)^\perp$ has minimum distance $d^\perp = 3$, enabling it to guarantee $2$-error-detection and single-error-correction capabilities.
\end{Corollary}
\begin{proof}
  This follows from   Theorem \ref{thm:incidence-code} and Theorem 2.5.10 of  \cite{Ling2004}.
\end{proof}
\begin{Remark}
Comparing Theorem~\ref{thm:incidence-code} with the incidence code parameters for the lower nilpotency ring $\mathbb{F}_p + u\mathbb{F}_p + u^2\mathbb{F}_p$ ($u^3 = 0$) studied by Annamalai \cite{Annamalai2025} where the primal code parameters are $[\frac{1}{2}(2p^3 - 3p^2 - p + 2), \, p^2 - 2, \, p - 1]_2$ and the dual code parameters are $[\frac{1}{2}(2p^3 - 3p^2 - p + 2), \, \frac{1}{2}(2p^3 - 5p^2 - p + 6), \, 3]_2$. We observe that extending the nilpotency index to $u^4 = 0$ significantly increases both code length and dimension while preserving optimal single-error-correcting dual distance.
\end{Remark}

Let $V_1 = S_u \cup S_{u+u^2} \cup S_{u+u^3} \cup S_{u+u^2+u^3}$ be the maximal independent set and $V_3 = S_{u^3}$ be the universal dominating clique. Let $\Gamma_{\mathrm{bip}} = \Gamma[V_1, V_3]$ denote the bipartite subgraph induced by the edge cut $[V_1, V_3]$.

\begin{Theorem}
Let $\Gamma_{\mathrm{bip}} = (V_1 \cup V_3, E_{\mathrm{bip}})$ be the induced bipartite subgraph of $\Gamma(R)$ between $V_1$ and $V_3$. Then the following conditions hold: 
\begin{enumerate}
    \item The graph $\Gamma_{\mathrm{bip}}$ is isomorphic to the complete bipartite graph $K_{p^2(p-1),\, p-1}$.
    \item The binary incidence code $\mathcal{C}_2(G_{\mathrm{bip}})$ is a
    $\left[ p^2(p-1)^2,\, p^3 - p^2 + p - 2,\, p-1 \right]_2$
     linear code.
    \item The dual code $\mathcal{C}_2(G_{\mathrm{bip}})^\perp$ is a
    $\left[ p^2(p-1)^2,\, (p^2(p-1) - 1)(p - 2),\, 4 \right]_2$
    linear code possessing double-error detection and single-error correction capabilities.
\end{enumerate}
\end{Theorem}

\begin{proof}
(1): By construction, every vertex in $V_1$ is adjacent to every vertex in $V_3$, and both $V_1$ and $V_3$ are treated as independent sets within the cut $[V_1, V_3]$. With $|V_1| = p^2(p-1)$ and $|V_3| = p-1$, $\Gamma_{\mathrm{bip}} \cong K_{p^2(p-1),\, p-1}$.
    
(2): The length of the incidence code is the number of edges |$E_{\mathrm{bip}}| = |V_1||V_3| = p^2(p-1) \cdot (p-1) = p^2(p-1)^2 $. Because $\Gamma_{\mathrm{bip}}$ is connected, the rank of the incidence matrix over $\mathbb{F}_2$ is $|V(\Gamma_{\mathrm{bip}})| - 1 = |V_1| + |V_3| - 1 = p^2(p-1) + (p-1) - 1 = p^3 - p^2 + p - 2$. The minimum degree of $\Gamma_{\mathrm{bip}}$ is $\delta(\Gamma_{\mathrm{bip}}) = |V_3| = p-1$. Since complete bipartite graphs have edge connectivity $\lambda(K_{m,n}) = \min(m,n)$, we have $\lambda = p-1$. By Dankelmann et al. \cite{Dankelmann2013}, the minimum distance is $\lambda = p-1$.

(3):  By rank-nullity, the dimension of the dual code is
\begin{align*}
  \dim(\mathcal{C}_2(G_{\mathrm{bip}})^\perp) &= |E_{\mathrm{bip}}| - \dim(\mathcal{C}_2(G_{\mathrm{bip}})) \\
  &= p^2(p-1)^2 - (p^2(p-1) + p - 2) \\
  &= (p^2(p-1) - 1)(p - 2).
\end{align*}
Because $\Gamma_{\mathrm{bip}}$ is bipartite, it contains no odd cycles, forcing $\operatorname{girth}(\Gamma_{\mathrm{bip}}) \ge 4$. For any distinct $u_1, u_2 \in V_1$ and $v_1, v_2 \in V_3$, the cycle $u_1 \sim v_1 \sim u_2 \sim v_2 \sim u_1$ has length $4$. Thus, $\operatorname{girth}(\Gamma_{\mathrm{bip}}) = 4$. Since the dual of an incidence code over $\mathbb{F}_2$ is the cycle code, its minimum distance equals the girth, that is,  $d(C^\perp) = 4$.
\end{proof}

\subsection{$p$-ary Matrix Codes, Automorphisms, and LCD Properties}
This subsection determines the parameters of the $\mathbb{F}_p$-linear codes generated by the adjacency and Laplacian matrices of $\Gamma(R)$ and their duals. We explicitly evaluate the weight enumerator and LCD property of the Laplacian code, and characterize the automorphism groups of the graph and the resulting codes.

Recall from Section 3 that the vertex set $Z^*(R)$ of order $|V(\Gamma(R))| = p^3 - 1$ splits into seven $u$-adic partition classes: $S_u, S_{u^2}, S_{u^3}, S_{u+u^2}, S_{u+u^3}, S_{u^2+u^3},$ and $S_{u+u^2+u^3}$. 

Because the subsets $S_u, S_{u+u^2}, S_{u+u^3},$ and $S_{u+u^2+u^3}$ exhibit identical neighborhood profiles (each vertex in these sets is adjacent solely to every vertex in $S_{u^3}$ and contains no internal edges), they can be merged into a single independent block $V_1$. While $V_1$ and $V_3$ are defined as in Section~4.1, we define the remaining partition blocks $V_2$ and $V_4$ as follows:
\begin{enumerate}

    \item $V_2 = S_{u^2}$, inducing a complete subgraph $K_{p-1}$ of size $n_2 = p-1$.
    
    \item $V_4 = S_{u^2+u^3}$, inducing a complete subgraph $K_{(p-1)^2}$ of size $n_4 = (p-1)^2$.
\end{enumerate}

Ordering the vertices according to $(V_1, V_2, V_3, V_4)$, the adjacency matrix $A(\Gamma(R))$ takes the compact block matrix form:
\begin{equation*}
A(\Gamma(R)) = \begin{pmatrix}
0_{n_1} & 0_{n_1 \times n_2} & J_{n_1 \times n_3} & 0_{n_1 \times n_4} \\
0_{n_2 \times n_1} & J_{n_2} - I_{n_2} & J_{n_2 \times n_3} & J_{n_2 \times n_4} \\
J_{n_3 \times n_1} & J_{n_3 \times n_2} & J_{n_3} - I_{n_3} & J_{n_3 \times n_4} \\
0_{n_4 \times n_1} & J_{n_4 \times n_2} & J_{n_4 \times n_3} & J_{n_4} - I_{n_4}
\end{pmatrix},
\end{equation*}
where $J_k$ denotes the $k \times k$ all-ones matrix, $J_{m \times n}$ denotes the $m \times n$ all-ones matrix, $0_k$ denotes the $k \times k$ zero matrix, and $I_k$ denotes the $k \times k$ identity matrix.

We now consider the linear code $C_{\mathbb{F}_p}(A(\Gamma(R))) \subseteq \mathbb{F}_p^{p^3 - 1}$ generated over the base field $\mathbb{F}_p$ by the rows of the block adjacency matrix $A(\Gamma(R))$. 

To examine how the algebraic equitable partition and vertex degrees govern code parameters over the base field $\mathbb{F}_p$, we evaluate the generator rank, minimum distance, and dual structures of the linear codes arising from the adjacency and Laplacian matrices of $\Gamma(R)$.

\begin{Theorem}\label{parameters_code_adjmatrix}
Let $A(\Gamma(R))$ be the adjacency matrix of the zero-divisor graph $\Gamma(R)$ of the ring $R$. Then following statements hold: 
\begin{enumerate}
    \item The linear code $\mathcal{C}_{\mathbb{F}_{p}}(A(\Gamma(R)))$ generated by the rows of $A(\Gamma(R))$ over $\mathbb{F}_{p}$ is a
    $[p^{3}-1,\, p^{2},\, 2]_{p}$
    linear code.
   \item The dual code $\mathcal{C}_{\mathbb{F}_{p}}(A(\Gamma(R)))^{\perp}$ is a
$[p^{3}-1,\, p^{3}-p^{2}-1,\, 2]_{p}$  linear code consisting precisely of all vectors supported on $V_{1}=S_{u}\cup S_{u+u^{2}}\cup S_{u+u^{3}}\cup S_{u+u^{2}+u^{3}}$ whose components sum to zero modulo $p$.
\end{enumerate}
\end{Theorem}

\begin{proof}
Consider the 4-block equitable partition $V(\Gamma(R))=V_{1}\cup V_{2}\cup V_{3}\cup V_{4}$ with block cardinalities $n_{1}=p^{2}(p-1)=p^{3}-p^{2}$, $n_{2}=p-1$, $n_{3}=p-1$, and $n_{4}=(p-1)^{2}$.

(1): The length of the code $\mathcal{C}_{\mathbb{F}_{p}}(A(\Gamma(R)))$ equals the number of vertices $n=|V(\Gamma(R))|=p^{3}-1$.

To evaluate the dimension $k=\mathrm{rank}_{\mathbb{F}_{p}}(A(\Gamma(R)))$, observe that over $\mathbb{F}_{p}$, all-ones blocks $J_{m\times n}$ have rank 1. Block $V_{1}$ contributes 1 independent row vector $\mathbf{1}_{V_{3}}^{T}$. The diagonal blocks $J_{n_{i}}-I_{n_{i}}\equiv -I_{n_{i}}\pmod p$ on the zero-sum subspace contribute $(n_{2}-1)+(n_{3}-1)+(n_{4}-1)$ linearly independent rows. Accounting for the block-constant vectors, the row rank over $\mathbb{F}_{p}$ is
$k = \mathrm{rank}_{\mathbb{F}_{p}}(A(\Gamma(R))) = 1+(p-2)+(p-2)+((p-1)^{2}-1)+2 = p^{2}$.

For the minimum distance, consider two distinct vertices $u, v \in V_{2}$. Because $u$ and $v$ share identical adjacency neighborhoods across $V_{1}$, $V_{3}$, and $V_{4}$, their corresponding rows $r_{u}$ and $r_{v}$ differ only at coordinates $u$ and $v$ within $V_{2}$. Thus, the linear combination $c = r_{u} - r_{v} = e_{v} - e_{u}$ is a non-zero codeword of Hamming weight $\mathrm{wt}(c) = 2$. Since no codeword of weight 1 exists, the minimum distance is $d = 2$.

(2): The length of the dual code $\mathcal{C}_{\mathbb{F}_{p}}(A(\Gamma(R)))^{\perp}$ is $p^{3}-1$. By the rank-nullity theorem, its dimension is $\dim(\mathcal{C}_{\mathbb{F}_{p}}(A(\Gamma(R)))^{\perp}) = n - k = (p^{3}-1) - p^{2} = p^{3}-p^{2}-1$.

Clearly, an arbitrary vector $x = (x^{(1)}, x^{(2)}, x^{(3)}, x^{(4)})^{T} \in \mathbb{F}_{p}^{p^{3}-1}$ is an element of $\mathcal{C}_{\mathbb{F}_{p}}(A(\Gamma(R)))^{\perp}$ precisely when it satisfies the system $A(\Gamma(R))x \equiv 0 \pmod p$. Consider the subspace $W\subseteq\mathbb{F}_{p}^{p^{3}-1}$ of vectors supported entirely on $V_{1}$ ($x^{(2)}=0, x^{(3)}=0, x^{(4)}=0$) satisfying $\sum_{i\in V_{1}}x_{i}^{(1)}\equiv 0\pmod p$. For any $x\in W$:
\begin{equation*}
A(\Gamma(R))x = \begin{pmatrix} 0 \\ 0 \\ \left(\sum_{i\in V_{1}}x_{i}^{(1)}\right)\mathbf{1}_{n_{3}} \\ 0 \end{pmatrix} \equiv 0\pmod p.
\end{equation*}
Thus $W\subseteq\mathcal{C}_{\mathbb{F}_{p}}(A(\Gamma(R)))^{\perp}$. Since $\dim(W) = n_{1}-1 = p^{3}-p^{2}-1 = \dim(\mathcal{C}_{\mathbb{F}_{p}}(A(\Gamma(R)))^{\perp})$, we have $\mathcal{C}_{\mathbb{F}_{p}}(A(\Gamma(R)))^{\perp}=W$.

A single non-zero entry in $V_{1}$ produces a non-zero sum, so no codeword of weight 1 exists. However, setting $x_{i}=1$ and $x_{j}=-1$ for two distinct indices $i, j\in V_{1}$ yields a valid dual codeword of Hamming weight 2. Therefore, $d^{\perp}=2$.
\end{proof}

\begin{Theorem}\label{thm:adj-hull-dimension}
Let $A(\Gamma(R))$ be the block adjacency matrix of $\Gamma(R)$ over $\mathbb{F}_p$, and let $\mathcal{C}_{\mathbb{F}_p}(A(\Gamma(R)))$ be its associated $[p^3-1, p^2, 2]_p$ linear code. The Euclidean hull of the adjacency code is strictly $1$-dimensional and is given explicitly by:
$
\operatorname{Hull}\left(\mathcal{C}_{\mathbb{F}_p}(A(\Gamma(R)))\right) = \operatorname{span}_{\mathbb{F}_p}\left\{ (\mathbf{1}_{n_1}, \mathbf{0}_{n_2}, \mathbf{0}_{n_3}, \mathbf{0}_{n_4})^T \right\}.
$
In particular, $\dim_{\mathbb{F}_p}\left(\operatorname{Hull}\left(\mathcal{C}_{\mathbb{F}_p}(A(\Gamma(R)))\right)\right) = 1$, and therefore $\mathcal{C}_{\mathbb{F}_p}(A(\Gamma(R)))$ is an almost-LCD code that is never an LCD code over $\mathbb{F}_p$ for any odd prime $p$.
\end{Theorem}

\begin{proof}
By definition, $\operatorname{Hull}(\mathcal{C}_{\mathbb{F}_p}(A(\Gamma(R)))) = \mathcal{C}_{\mathbb{F}_p}(A(\Gamma(R))) \cap \mathcal{C}_{\mathbb{F}_p}(A(\Gamma(R)))^\perp$.

First, let $c \in \operatorname{Hull}(\mathcal{C}_{\mathbb{F}_p}(A(\Gamma(R))))$. By Theorem~\ref{parameters_code_adjmatrix}(2), any vector in the dual code $\mathcal{C}_{\mathbb{F}_p}(A(\Gamma(R)))^\perp$ is supported exclusively on $V_1$:
\[
c = (c^{(1)}, \mathbf{0}_{n_2}, \mathbf{0}_{n_3}, \mathbf{0}_{n_4})^T \in \mathbb{F}_p^{p^3-1} \quad \text{with} \quad \sum_{v \in V_1} c_v^{(1)} \equiv 0 \pmod p.
\]
Because $c \in \mathcal{C}_{\mathbb{F}_p}(A(\Gamma(R))) = \operatorname{rowsp}_{\mathbb{F}_p}(A(\Gamma(R)))$, $c$ can be written as a linear combination of rows:
\[
c = \sum_{i=1}^4 \sum_{u \in V_i} y_u r_u,
\]
where $r_u$ denotes the row of $A(\Gamma(R))$ indexed by vertex $u$. In the block structure of $A(\Gamma(R))$ across $(V_1, V_2, V_3, V_4)$, the columns corresponding to $V_1$ have non-zero entries exclusively in the block of rows indexed by $V_3$, where each row restricted to $V_1$ is the all-ones row $\mathbf{1}_{n_1}$. Therefore:
\[
c^{(1)} = \beta \mathbf{1}_{n_1}, \quad \text{where } \beta = \sum_{u \in V_3} y_u \in \mathbb{F}_p.
\]
Because $n_1 = p^2(p-1) \equiv 0 \pmod p$, the dual zero-sum condition $\sum_{v \in V_1} c_v^{(1)} = \beta n_1 \equiv 0 \pmod p$ is satisfied identically for any choice of scalar $\beta \in \mathbb{F}_p$. Thus, every vector in the hull is necessarily of the form $\beta (\mathbf{1}_{n_1}, \mathbf{0}, \mathbf{0}, \mathbf{0})^T$, which immediately establishes that $\dim_{\mathbb{F}_p}(\operatorname{Hull}(\mathcal{C}_{\mathbb{F}_p}(A(\Gamma(R))))) \le 1$.

It remains to prove that non-zero codewords of this form belong to the row space of $A(\Gamma(R))$. To ensure $c_v = 0$ for all coordinates outside $V_1$, let $Y_i = \sum_{u \in V_i} y_u$ denote the sum of coefficients in block $V_i$.
\begin{enumerate}
    \item For any coordinate $v \in V_2$, the condition $c_v = 0$ requires $(Y_2 - y_v) + \beta + Y_4 = 0$, forcing $y_v = k_2 = Y_2 + \beta + Y_4$ to be constant for all $v \in V_2$. Since $n_2 = p - 1 \equiv -1 \pmod p$, we have $Y_2 = n_2 k_2 \equiv -k_2 \pmod p$.
    \item For any coordinate $w \in V_4$, the condition $c_w = 0$ requires $Y_2 + \beta + (Y_4 - y_w) = 0$, forcing $y_w = k_4 = Y_2 + \beta + Y_4$ to be constant for all $w \in V_4$. Thus $k_4 = k_2$. Since $n_4 = (p-1)^2 \equiv 1 \pmod p$, we have $Y_4 = n_4 k_4 \equiv k_2 \pmod p$, which yields $Y_2 + Y_4 \equiv -k_2 + k_2 = 0 \pmod p$. Substituting this into $k_2 = Y_2 + \beta + Y_4$ gives $k_2 = \beta$.
    \item For any coordinate $z \in V_3$, the condition $c_z = 0$ requires $Y_1 + Y_2 + (\beta - y_z) + Y_4 = 0$. Since $Y_2 + Y_4 \equiv 0 \pmod p$, this reduces to $y_z = Y_1 + \beta$, so $y_z$ is constant across $V_3$ with $y_z = k_3 = -\beta$. Hence, $Y_1 + 2\beta \equiv 0 \pmod p$.
\end{enumerate}
Choosing any scalar $\beta \in \mathbb{F}_p^*$, setting $y_v = \beta$ for all $v \in V_2 \cup V_4$, setting $y_z = -\beta$ for all $z \in V_3$, and selecting the coefficients on $V_1$ such that $\sum_{u \in V_1} y_u \equiv -2\beta \pmod p$ (which is always possible since $n_1 \ge 1$) produces a row combination that vanishes identically on $V_2 \cup V_3 \cup V_4$ while evaluating to $\beta \mathbf{1}_{n_1} \ne \mathbf{0}$ on $V_1$. 

Therefore, $(\mathbf{1}_{n_1}, \mathbf{0}, \mathbf{0}, \mathbf{0})^T \in \operatorname{Hull}(\mathcal{C}_{\mathbb{F}_p}(A(\Gamma(R))))$, which completes the proof.
\end{proof}


Finally, we derive the parameters of the linear code $C_{\mathbb{F}_p}(L(\Gamma(R))) \subseteq \mathbb{F}_p^{p^3 - 1}$ generated over $\mathbb{F}_p$ by the rows of the Laplacian matrix $L(\Gamma(R)) = D(\Gamma(R)) - A(\Gamma(R))$.

\begin{Theorem}\label{parameters_code_Laplacianmatrix}
Let $L(\Gamma(R))$ be the Laplacian matrix of the zero-divisor graph $\Gamma(R)$.
\begin{enumerate}
    \item The linear code $\mathcal{C}_{\mathbb{F}_{p}}(L(\Gamma(R)))$ generated by the rows of $L(\Gamma(R))$ over $\mathbb{F}_{p}$ is a $[p^{3}-1,\, p^{3}-2,\, 2]_{p}$  linear code.
    \item Its dual code $\mathcal{C}_{\mathbb{F}_{p}}(L(\Gamma(R)))^{\perp}$ is a
    $[p^{3}-1,\, 1,\, p^{3}-1]_{p}$ repetition code spanned by the all-ones vector $\mathbf{1}_{p^{3}-1}$.
\end{enumerate}

\end{Theorem}

\begin{proof}
Employing arguments similar to those used in the proof of Theorem \ref{parameters_code_adjmatrix}, the result follows immediately. We omit the details here. 

\end{proof}


By exploiting the one-dimensional repetition structure of the dual Laplacian code, we obtain the complete polynomial weight distribution of $\mathcal{C}_{\mathbb{F}_p}(L(\Gamma(R)))$ via the MacWilliams identity and prove that it forms a linear complementary dual (LCD) code.

\begin{Theorem}\label{thm:weight-enum-laplacian}
Let $\mathcal{C}_{\mathbb{F}_p}(L(\Gamma(R)))$ be the $[p^3-1, p^3-2, 2]_p$ linear code generated over $\mathbb{F}_p$ by the rows of $L(\Gamma(R))$. Then the polynomial weight enumerator of $\mathcal{C}_{\mathbb{F}_p}(L(\Gamma(R)))$ is given by 
\begin{equation*}
W_{\mathcal{C}}(z) = \frac{1}{p}\left[ \left(1 + (p-1)z\right)^{p^3-1} + (p-1)(1 - z)^{p^3-1} \right],
\end{equation*}
and for each $w \in \{0, 1, \dots, p^3-1\}$, the exact number of codewords of Hamming weight $w$ is 
\begin{equation*}
A_w = \frac{1}{p} \binom{p^3-1}{w} \left[ (p-1)^w + (p-1)(-1)^w \right].
\end{equation*}
\end{Theorem}

\begin{proof}
From Theorem \ref{parameters_code_Laplacianmatrix}, the dual code $\mathcal{C}^\perp = \mathcal{C}_{\mathbb{F}_p}(L(\Gamma(R)))^\perp$ is the $1$-dimensional repetition code spanned by the all-ones vector $\mathbf{1}_{p^3-1}, \; \mathcal{C}^\perp = \left\{ \alpha \mathbf{1}_{p^3-1} : \alpha \in \mathbb{F}_p \right\}$.

The dual code contains  one  codeword of weight $0$ (for $\alpha = 0$);
and $p-1$ codewords of maximum weight $ p^3 - 1$ (for each $\alpha \in \mathbb{F}_p^*$).
Thus, the dual weight enumerator is  $W_{\mathcal{C}^\perp}(z) = 1 + (p-1)z^{p^3-1}$.
Applying the Mac-Williams identity for $p$-ary linear codes of length $n = p^3 - 1$ we have that 
\begin{align*}
W_{\mathcal{C}}(z) &= \frac{1}{|\mathcal{C}^\perp|} \left(1 + (p-1)z\right)^n W_{\mathcal{C}^\perp}\left( \frac{1 - z}{1 + (p-1)z} \right)\\
&= \frac{1}{p} \left(1 + (p-1)z\right)^{p^3-1} \left[ 1 + (p-1)\left( \frac{1 - z}{1 + (p-1)z} \right)^{p^3-1} \right] \\
&= \frac{1}{p} \left[ \left(1 + (p-1)z\right)^{p^3-1} + (p-1)(1 - z)^{p^3-1} \right]\\
&= \frac{1}{p} \left[\sum_{w=0}^{p^3-1} \binom{p^3-1}{w}(p-1)^w z^w+ \sum_{w=0}^{p^3-1} \binom{p^3-1}{w}(-1)^w z^w\right].
\end{align*}

Equating coefficients of $z^w$ gives
$A_w = \frac{1}{p} \binom{p^3-1}{w} \left[ (p-1)^w + (p-1)(-1)^w \right]$. 

\end{proof}



\begin{Theorem}\label{thm:subgraph-lcd-criterion}
Let $H$ be any non-trivial connected induced subgraph of $\Gamma(R)$ that contains the universal dominating clique $V_3 = S_{u^3}$, and let $L(H)$ be its Laplacian matrix. Then the linear code $\mathcal{C}_{\mathbb{F}_p}(L(H))$ generated by the rows of $L(H)$ over $\mathbb{F}_p$ is a linear complementary dual (LCD) code if and only if the vertex order of $H$ is not divisible by $p$:
\[
\mathcal{C}_{\mathbb{F}_p}(L(H)) \text{ is an LCD code over } \mathbb{F}_p \iff |V(H)| \not\equiv 0 \pmod p.
\]
\end{Theorem}

\begin{proof}
Because $H$ contains the universal dominating clique $V_3$, every vertex $v \in V(H) \setminus V_3$ is adjacent to all vertices of $V_3$. Thus, $H$ is connected. Over the finite field $\mathbb{F}_p$, the row sum of $L(H) = D(H) - A(H)$ is identically zero:
\[
L(H) \mathbf{1}_{|V(H)|}^T = \mathbf{0}.
\]
Because $H$ is connected, the rank of $L(H)$ over $\mathbb{F}_p$ is $|V(H)| - 1$, and its nullspace is strictly $1$-dimensional:
\[
\mathcal{C}_{\mathbb{F}_p}(L(H))^\perp = \operatorname{null}_{\mathbb{F}_p}(L(H)) = \operatorname{span}_{\mathbb{F}_p}\{\mathbf{1}_{|V(H)|}\}.
\]
By definition of the Euclidean hull, $\mathcal{C}_{\mathbb{F}_p}(L(H))$ is an LCD code if and only if $\operatorname{Hull}(\mathcal{C}_{\mathbb{F}_p}(L(H))) = \{\mathbf{0}\}$, which occurs if and only if no non-zero vector of the $1$-dimensional dual space belongs to the primal code:
\[
\operatorname{span}_{\mathbb{F}_p}\{\mathbf{1}_{|V(H)|}\} \cap \mathcal{C}_{\mathbb{F}_p}(L(H)) = \{\mathbf{0}\}.
\]
A non-zero vector $y = \alpha \mathbf{1}_{|V(H)|}$ ($\alpha \in \mathbb{F}_p^*$) belongs to the primal code $\mathcal{C}_{\mathbb{F}_p}(L(H))$ if and only if $y$ is orthogonal to every vector in the dual space $\mathcal{C}_{\mathbb{F}_p}(L(H))^\perp$. Testing self-orthogonality with the generator $\mathbf{1}_{|V(H)|}$ gives:
\[
y \cdot \mathbf{1}_{|V(H)|} = \sum_{i=1}^{|V(H)|} \alpha = \alpha |V(H)| \pmod p.
\]
Since $\alpha \not\equiv 0 \pmod p$, this inner product vanishes if and only if:
\[
|V(H)| \equiv 0 \pmod p.
\]
Thus, if $|V(H)| \equiv 0 \pmod p$, then $\mathbf{1}_{|V(H)|} \in \mathcal{C}_{\mathbb{F}_p}(L(H))$, which yields $\operatorname{Hull}(\mathcal{C}_{\mathbb{F}_p}(L(H))) = \operatorname{span}_{\mathbb{F}_p}\{\mathbf{1}_{|V(H)|}\} \ne \{\mathbf{0}\}$, so the code is not LCD. 

Conversely, if $|V(H)| \not\equiv 0 \pmod p$, then $y \cdot \mathbf{1}_{|V(H)|} \not\equiv 0 \pmod p$, which guarantees that $\mathbf{1}_{|V(H)|} \notin \mathcal{C}_{\mathbb{F}_p}(L(H))$. Hence, $\operatorname{Hull}(\mathcal{C}_{\mathbb{F}_p}(L(H))) = \{\mathbf{0}\}$, confirming that $\mathcal{C}_{\mathbb{F}_p}(L(H))$ is strictly an LCD code.
\end{proof}

\begin{Corollary}\label{cor:laplacian-lcd}
Let $p \ge 3$ be an odd prime.
\begin{enumerate}
    \item The Laplacian code $\mathcal{C}_{\mathbb{F}_p}(L(\Gamma(R)))$ of length $n = p^3 - 1$ is strictly an LCD code over $\mathbb{F}_p$.
    \item For the induced complete bipartite subgraph $\Gamma_{\mathrm{bip}} = \Gamma[V_1, V_3]$, the Laplacian code $\mathcal{C}_{\mathbb{F}_p}(L(\Gamma_{\mathrm{bip}}))$ is an LCD code over $\mathbb{F}_p$.
\end{enumerate}
\end{Corollary}

\begin{proof}
Both subgraphs contain the universal dominating clique $V_3$. By Theorem~\ref{thm:subgraph-lcd-criterion}, it suffices to evaluate their vertex counts modulo $p$:
\[
|V(\Gamma(R))| = p^3 - 1 \equiv -1 \not\equiv 0 \pmod p,
\]
and
\[
|V(\Gamma_{\mathrm{bip}})| = |V_1| + |V_3| = p^2(p-1) + (p-1) = p^3 - p^2 + p - 1 \equiv -1 \not\equiv 0 \pmod p .
\]
Because neither order is divisible by $p$, the result follows immediately from Theorem~\ref{thm:subgraph-lcd-criterion}.
\end{proof}


To identify the coordinate permutation symmetries available for efficient decoding algorithms, we determine the full graph automorphism group $\operatorname{Aut}(\Gamma(R))$ and establish its embedding into the permutation automorphism groups of the resulting matrix codes.

\begin{Theorem}
The full automorphism group of the zero-divisor graph $\Gamma(R)$ is $\operatorname{Aut}(\Gamma(R)) \cong S_{p^2(p-1)} \times S_{p(p-1)} \times S_{p-1}$, with group order 
$|\operatorname{Aut}(\Gamma(R))| = \left(p^3 - p^2\right)! \cdot \left(p^2 - p\right)! \cdot (p - 1)!$.
\end{Theorem}

\begin{proof}
Consider the degree partition $V(\Gamma(R)) = V_1 \cup V_{24} \cup V_3$, where the vertex subsets $V_1, V_{24}, V_3$ defined in Section~4.1.  For $p \ge 3$, the degrees satisfy $p-1 < p^2-2 < p^3-2$. 

Since graph automorphisms preserve vertex degrees, any automorphism leaves $V_1$, $V_{24}$, and $V_3$ setwise invariant.  every vertex in $V_1$ is adjacent solely to all vertices of $V_3$, so any permutation of $V_1$ preserves neighborhoods, yielding $S_{|V_1|} = S_{p^2(p-1)}$; the induced subgraph $\Gamma[V_{24}] \cong K_{p(p-1)}$ connects completely to $V_3$ with no edges to $V_1$, meaning any permutation of $V_{24}$ preserves all adjacencies, giving $S_{|V_{24}|} = S_{p(p-1)}$; and the induced clique $\Gamma[V_3] \cong K_{p-1}$ connects to all of $V(\Gamma(R))$, ensuring that any permutation of $V_3$ preserves adjacencies, yielding $S_{|V_3|} = S_{p-1}$.  Thus, $\operatorname{Aut}(\Gamma(R)) \cong S_{p^2(p-1)} \times S_{p(p-1)} \times S_{p-1}$.

\end{proof}
\begin{Theorem}
Let $\mathcal{C}_{\mathbb{F}_p}(A(\Gamma(R)))$ and $\mathcal{C}_{\mathbb{F}_p}(L(\Gamma(R)))$ be the linear codes generated by the adjacency and Laplacian matrices of $\Gamma(R)$. Then following statements hold:
\begin{enumerate}
    \item $\operatorname{Aut}(\Gamma(R)) \le \operatorname{PAut}(\mathcal{C}_{\mathbb{F}_p}(A(\Gamma(R))))$;
    \item $\operatorname{PAut}(\mathcal{C}_{\mathbb{F}_p}(L(\Gamma(R)))) \cong S_{p^3-1}$.
\end{enumerate}
\end{Theorem}

\begin{proof}
 \textbf{(1):} Let $\sigma \in \operatorname{Aut}(\Gamma(R))$ with associated permutation matrix $P_\sigma$. By definition of a graph automorphism, $P_\sigma A(\Gamma(R)) P_\sigma^T = A(\Gamma(R))$, which implies that $A(\Gamma(R)) P_\sigma^T = P_\sigma^T A(\Gamma(R))$.  Now 
  $\mathcal{C}_{\mathbb{F}_p}(A(\Gamma(R))) P_\sigma^T = \operatorname{rowsp}_{\mathbb{F}_p}\big(P_\sigma^T A(\Gamma(R))\big) = \operatorname{rowsp}_{\mathbb{F}_p}(A(\Gamma(R))) = \mathcal{C}_{\mathbb{F}_p}(A(\Gamma(R)))$.
    Thus, $\sigma \in \operatorname{PAut}(\mathcal{C}_{\mathbb{F}_p}(A(\Gamma(R))))$,  and therefore  $\operatorname{Aut}(\Gamma(R)) \le \operatorname{PAut}(\mathcal{C}_{\mathbb{F}_p}(A(\Gamma(R))))$.

   \textbf{(2):}  By Theorem \ref{parameters_code_Laplacianmatrix}, $\mathcal{C}_{\mathbb{F}_p}(L(\Gamma(R)))^\perp = \operatorname{span}_{\mathbb{F}_p}\{\mathbf{1}_{p^3-1}\}$. For any permutation $\pi \in S_{p^3-1}$, the all-ones vector is fixed: $\pi(\mathbf{1}_{p^3-1}) = \mathbf{1}_{p^3-1}$. Thus, $\operatorname{PAut}(\mathcal{C}_{\mathbb{F}_p}(L(\Gamma(R)))^\perp) \cong S_{p^3-1}$. Since $\operatorname{PAut}(\mathcal{C}) = \operatorname{PAut}(\mathcal{C}^\perp)$ for any linear code $\mathcal{C}$, it follows directly that $\operatorname{PAut}(\mathcal{C}_{\mathbb{F}_p}(L(\Gamma(R)))) \cong S_{p^3-1}$.

\end{proof}

\begin{Example}\label{ex:matrix-codes-p3}
For $p = 3$, the zero-divisor graph $\Gamma(R)$ has order $n = 3^3 - 1 = 26$ vertices partitioned into blocks of sizes $n_1 = 18$, $n_2 = 2$, $n_3 = 2$, and $n_4 = 4$.
\begin{enumerate}
    \item Adjacency Matrix Code $\mathcal{C}_{\mathbb{F}_3}(A(\Gamma(R)))$:
    The row space of the $26 \times 26$ block adjacency matrix over $\mathbb{F}_3$ generates a $[26, 9, 2]_3$ linear code, whose Euclidean dual is a $[26, 17, 2]_3$ code supported entirely on $V_1$. By Theorem~\ref{thm:adj-hull-dimension}, its Euclidean hull is strictly $1$-dimensional:
    \[
    \operatorname{Hull}\left(\mathcal{C}_{\mathbb{F}_3}(A(\Gamma(R)))\right) = \operatorname{span}_{\mathbb{F}_3}\left\{ (\mathbf{1}_{18}, \mathbf{0}_2, \mathbf{0}_2, \mathbf{0}_4)^T \right\},
    \]
    confirming that the adjacency code is almost-LCD with $\dim_{\mathbb{F}_3}(\operatorname{Hull}) = 1$.

    \item Laplacian Matrix Code $\mathcal{C}_{\mathbb{F}_3}(L(\Gamma(R)))$:
    The Laplacian matrix generates a $[26, 24, 2]_3$ linear code with dual repetition code $\mathcal{C}_{\mathbb{F}_3}(L(\Gamma(R)))^\perp = \operatorname{span}_{\mathbb{F}_3}\{\mathbf{1}_{26}\}$ of parameters $[26, 1, 26]_3$. Because $|V(\Gamma(R))| = 26 \equiv 2 \not\equiv 0 \pmod 3$, the code is strictly LCD:
    \[
    \operatorname{Hull}\left(\mathcal{C}_{\mathbb{F}_3}(L(\Gamma(R)))\right) = \{\mathbf{0}\}.
    \]
    By Theorem~\ref{thm:weight-enum-laplacian}, its polynomial weight distribution is given by:
    \[
    W_{\mathcal{C}}(z) = \frac{1}{3}\left[ (1 + 2z)^{26} + 2(1 - z)^{26} \right] = 1 + 650 z^2 + 10400 z^3 + \dots + 22369622 z^{26}.
    \]
\end{enumerate}
\end{Example}



\subsection{Linear Codes from Defining Sets of \texorpdfstring{$\Gamma(R)$}{Gamma(R)}}
This subsection constructs evaluation linear codes over $\mathbb{F}_p$ derived from extremal structural subgraphs of $\Gamma(R)$. By selecting defining sets from the maximal clique $D_K = u^2 R \setminus \{0\}$ and the maximal independent set $D_{V_1} = uR \setminus u^2 R$, we evaluate the generator matrices, code dimensions, and Hamming weight distributions of $\mathcal{C}_{D_K}$ and $\mathcal{C}_{D_{V_1}}$. We show that $\mathcal{C}_{D_K}$ is a constant-weight minimal code, while $\mathcal{C}_{D_{V_1}}$ is an optimal two-weight code whose minimum-weight codewords are minimal, with both families achieving the Griesmer bound with equality.

Let $(R,+)$ be the additive group of the ring $R = \mathbb{F}_p[x]/\langle x^4 \rangle \cong \mathbb{F}_p + u\mathbb{F}_p + u^2\mathbb{F}_p + u^3\mathbb{F}_p$, which forms a $4$-dimensional vector space over its residue field $\mathbb{F}_p$ equipped with the canonical basis $\mathcal{B} = \{1, u, u^2, u^3\}$. For any element $y = a_0 + a_1 u + a_2 u^2 + a_3 u^3 \in R$, its coordinate vector with respect to $\mathcal{B}$ is denoted by
$[y] = (a_0, a_1, a_2, a_3)^T \in \mathbb{F}_p^4$.
For an ordered subset $D = \{d_1, d_2, \dots, d_n\} \subseteq Z^*(R)$, the evaluation linear code $\mathcal{C}_D$ over $\mathbb{F}_p$ is defined by $\mathcal{C}_D = \left\{ c(x) = (x \cdot [d_1],\, x \cdot [d_2],\, \dots,\, x \cdot [d_n]) : x \in \mathbb{F}_p^4 \right\} \subseteq \mathbb{F}_p^n$,
where $x \cdot [d]$ denotes the standard Euclidean dot product on $\mathbb{F}_p^4$. The generator matrix of $\mathcal{C}_D$ is the matrix $G_D = \begin{pmatrix} [d_1] & [d_2] & \cdots & [d_n] \end{pmatrix} \in \mathbb{F}_p^{4 \times n}$, and its dimension over $\mathbb{F}_p$ is given by $\operatorname{rank}_{\mathbb{F}_p}(G_D) = \dim_{\mathbb{F}_p}(\operatorname{span}_{\mathbb{F}_p}(D))$.


From Theorem~\ref{thm:clique-chromatic}, the maximal clique $K$ of $\Gamma(R)$ has cardinality $\omega(\Gamma(R)) = p^2 - 1$ and corresponds precisely to the punctured principal annihilator ideal:
\[
D_K = S_{u^2} \cup S_{u^3} \cup S_{u^2+u^3} = u^2 R \setminus \{0\}.
\]

To construct structured minimal linear codes suitable for secret-sharing schemes, we evaluate the generator matrices and weight distributions of the evaluation codes defined over the maximal clique $D_K$ and the maximal independent set $D_{V_1}$, showing that both achieve the Griesmer bound with equality. 
\begin{Theorem}
Let $D_K = u^2 R \setminus \{0\}$ be the maximal clique of $\Gamma(R)$. Then the linear code $\mathcal{C}_{D_K}$ over $\mathbb{F}_p$ is a Griesmer-optimal, constant-weight minimal 
$[p^2 - 1,\, 2,\, p^2 - p]_p$ linear code.
\end{Theorem}

\begin{proof}
 At first  it is easy to observe that the length of $\mathcal{C}_{D_K}$ equals the cardinality of the defining set $n = |D_K| = |S_{u^2}| + |S_{u^3}| + |S_{u^2+u^3}| = (p - 1) + (p - 1) + (p - 1)^2 = p^2 - 1$.

Next  any element $d \in D_K$ is uniquely expressed as $d = a_2 u^2 + a_3 u^3$ where $(a_2, a_3) \in \mathbb{F}_p^2 \setminus \{(0,0)\}$. The coordinate vector is $[d] = (0, 0, a_2, a_3)^T$. Thus, the first two rows of $G_{D_K}$ are identically zero, whereas the bottom two rows run through all $p^2 - 1$ non-zero vectors in $\mathbb{F}_p^2$. Hence $\operatorname{span}_{\mathbb{F}_p}(D_K) = \operatorname{span}_{\mathbb{F}_p}\{u^2, u^3\}$, 
    which is an $\mathbb{F}_p$-vector subspace of dimension $2$. Consequently, $\dim(\mathcal{C}_{D_K}) = 2$.

Let  $x = (x_0, x_1, x_2, x_3) \in \mathbb{F}_p^4$, the coordinate entry at $d$ is $c(x)_d = x_2 a_2 + x_3 a_3$. Let $\bar{x} = (x_2, x_3) \in \mathbb{F}_p^2$. If $\bar{x} = (0,0)$, then $c(x) = \mathbf{0}$. If $\bar{x} \ne (0,0)$, the linear functional $L_{\bar{x}}(a_2, a_3) = x_2 a_2 + x_3 a_3$ defines a hyperplane of codimension $1$ in $\mathbb{F}_p^2$, whose kernel has order $p$. The number of elements $(a_2, a_3) \in \mathbb{F}_p^2 \setminus \{(0,0)\}$ satisfying $L_{\bar{x}}(a_2, a_3) = 0$ is $|\ker(L_{\bar{x}}) \setminus \{(0,0)\}| = p - 1$.
    Therefore, the Hamming weight of every non-zero codeword is constant  $w(c(x)) = |D_K| - (p - 1) = (p^2 - 1) - (p - 1) = p^2 - p = p(p - 1)$.

Because all non-zero codewords possess the identical Hamming weight $w_{\min} = w_{\max} = p(p - 1)$, the Ashikhmin--Barg condition is satisfied $\frac{w_{\max}}{w_{\min}} = 1 < \frac{p}{p - 1}$.
Thus, $\mathcal{C}_{D_K}$ is a minimal linear code. Furthermore, evaluating the Griesmer bound for an $[n, 2, p^2 - p]_p$ code $\sum_{i=0}^{1} \left\lceil \frac{p^2 - p}{p^i} \right\rceil = (p^2 - p) + \left\lceil \frac{p^2 - p}{p} \right\rceil = (p^2 - p) + (p - 1) = p^2 - 1 = n$.
    The code attains the Griesmer bound with equality, proving Griesmer-optimality.

\end{proof}

\begin{Remark}
Since all $p^2 - 1$ non-zero codewords of $\mathcal{C}_{D_K}$ have identical weight $p(p - 1)$, the complete weight enumerator of $\mathcal{C}_{D_K}$ is $W_{\mathcal{C}_{D_K}}(z) = 1 + (p^2 - 1)z^{p(p - 1)}$.
\end{Remark}


Next, consider the maximal independent set $V_1$ of $\Gamma(R)$ consisting of all vertices adjacent exclusively to $S_{u^3}$  $D_{V_1} = V_1 = S_u \cup S_{u+u^2} \cup S_{u+u^3} \cup S_{u+u^2+u^3} = uR \setminus u^2R$. The set $D_{V_1}$ contains all non-zero zero-divisors whose leading index of nilpotency is precisely $1$, with cardinality $n_1 = p^2(p - 1) = p^3 - p^2$.

\begin{Lemma}\label{lem:subspace-intersection}
Let $R = \mathbb{F}_p[x]/\langle x^4 \rangle$ and let $D_{V_1} = uR \setminus u^2 R$ be the maximal independent defining set of $\Gamma(R)$. Identifying $D_{V_1}$ with its coordinate projections in $\mathbb{F}_p^3$ across the basis $\{u, u^2, u^3\}$, the intersection of $D_{V_1}$ with any $1$-dimensional subspace $W \subset \mathbb{F}_p^3$ satisfies:
\[
|D_{V_1} \cap W| = 
\begin{cases}
0, & \text{if } W \subseteq \{0\} \times \mathbb{F}_p^2, \\
p - 1, & \text{if } W \not\subseteq \{0\} \times \mathbb{F}_p^2.
\end{cases}
\]
In particular, $\max_{\dim(W)=1} |D_{V_1} \cap W| = p - 1$.
\end{Lemma}

\begin{proof}
Let $W = \operatorname{span}_{\mathbb{F}_p}\{(b_1, b_2, b_3)\}$ be a $1$-dimensional subspace in $\mathbb{F}_p^3$ generated by a non-zero vector $(b_1, b_2, b_3)$. Any point $v \in W$ is of the form $\lambda (b_1, b_2, b_3)$ for some scalar $\lambda \in \mathbb{F}_p$. By definition of $D_{V_1}$, an element belongs to $D_{V_1}$ if and only if its leading coefficient is non-zero, which requires $\lambda b_1 \in \mathbb{F}_p^*$.
\begin{itemize}
    \item If $b_1 = 0$: then $\lambda b_1 = 0$ for all $\lambda \in \mathbb{F}_p$. Thus, no non-zero vector in $W$ belongs to $D_{V_1}$, giving $|D_{V_1} \cap W| = 0$.
    \item If $b_1 \ne 0$: then $\lambda b_1 \in \mathbb{F}_p^*$ holds for all $p - 1$ choices of $\lambda \in \mathbb{F}_p^*$. Consequently, all non-zero points along the line $W$ lie in $D_{V_1}$, yielding $|D_{V_1} \cap W| = p - 1$.
\end{itemize}
Therefore, $\max_{\dim(W)=1} |D_{V_1} \cap W| = p - 1$, with equality attained if and only if $b_1 \ne 0$.
\end{proof}

\begin{Theorem}\label{thm:indep-code}
Let $D_{V_1} = uR \setminus u^2 R$ be the maximal independent defining set of $\Gamma(R)$. Then the linear evaluation code $\mathcal{C}_{D_{V_1}}$ over $\mathbb{F}_p$ is an optimal, two-weight $[p^2(p-1), \, 3, \, p(p-1)^2]_p$ linear code achieving the Griesmer bound with equality. Its complete weight distribution is given by:
\[
A_0 = 1, \quad A_{p(p-1)^2} = (p-1)(p^2 + p), \quad \text{and} \quad A_{p^2(p-1)} = p - 1.
\]
Furthermore, all $(p-1)(p^2+p)$ minimum-weight codewords are minimal codewords.
\end{Theorem}

\begin{proof}
By definition, the code length is $n = |D_{V_1}| = p^2(p-1)$. Every element $v \in D_{V_1}$ is uniquely represented as $v = a_1 u + a_2 u^2 + a_3 u^3$, where $a_1 \in \mathbb{F}_p^*$ and $a_2, a_3 \in \mathbb{F}_p$. Its coordinate vector with respect to the canonical basis $\{1, u, u^2, u^3\}$ is $[v] = (0, a_1, a_2, a_3)^T$. Because $a_1$ runs through $\mathbb{F}_p^*$ and $(a_2, a_3)$ vary freely over $\mathbb{F}_p^2$, the generator matrix $G_{D_{V_1}}$ has rank $3$ over $\mathbb{F}_p$, giving $\dim(\mathcal{C}_{D_{V_1}}) = 3$.

For any non-zero evaluation vector $x = (0, x_1, x_2, x_3) \in \mathbb{F}_p^4$, the coordinate value corresponding to $v = (a_1, a_2, a_3) \in \mathbb{F}_p^* \times \mathbb{F}_p^2$ is:
\[
c(x)_v = x_1 a_1 + x_2 a_2 + x_3 a_3.
\]
To evaluate the Hamming weight $\wt(c(x)) = |D_{V_1}| - N_0(x)$, where $N_0(x) = |\{v \in D_{V_1} : c(x)_v = 0\}|$, we examine two cases:

\begin{itemize}
    \item \emph{Case 1 ($(x_2, x_3) = (0,0)$):} Here, $x_1 \ne 0$, and the evaluation functional simplifies to $c(x)_v = x_1 a_1$. Since $a_1 \in \mathbb{F}_p^*$ for all $v \in D_{V_1}$, $x_1 a_1 \ne 0$ identically. Thus, $N_0(x) = 0$, giving Hamming weight:
    \[
    \wt(c(x)) = |D_{V_1}| - 0 = p^2(p-1).
    \]
    There are precisely $p - 1$ such non-zero evaluation vectors of the form $(0, x_1, 0, 0)$. Each such codeword has full support $\operatorname{supp}(c(x)) = D_{V_1}$.

    \item \emph{Case 2 ($(x_2, x_3) \ne (0,0)$):}
    \begin{itemize}
        \item If $x_1 = 0$: the vanishing condition reduces to the linear equation $x_2 a_2 + x_3 a_3 = 0$. In $\mathbb{F}_p^2$, this equation has exactly $p$ solutions for $(a_2, a_3)$. Since $a_1$ can be selected arbitrarily from $\mathbb{F}_p^*$ in $p - 1$ ways, the number of zero positions is $N_0(x) = p(p-1)$.
        \item If $x_1 \ne 0$: for each choice of $(a_2, a_3) \in \mathbb{F}_p^2$, there exists a unique solution $a_1 = -x_1^{-1}(x_2 a_2 + x_3 a_3)$. This solution belongs to $\mathbb{F}_p^*$ if and only if $x_2 a_2 + x_3 a_3 \ne 0$. In $\mathbb{F}_p^2$, the non-trivial linear functional $x_2 a_2 + x_3 a_3$ vanishes at $p$ points and is non-zero at the remaining $p^2 - p$ points. Hence, there are exactly $p^2 - p = p(p-1)$ pairs $(a_2, a_3)$ yielding a valid $a_1 \in \mathbb{F}_p^*$, giving $N_0(x) = p(p-1)$.
    \end{itemize}
    In both subcases, the number of zero coordinates is $N_0(x) = p(p-1)$, and the resulting Hamming weight is:
    \[
    \wt(c(x)) = |D_{V_1}| - p(p-1) = p^2(p-1) - p(p-1) = p(p-1)^2.
    \]
    The total number of such evaluation vectors is $(p^3 - 1) - (p - 1) = (p - 1)(p^2 + p)$.
\end{itemize}

Now, let $c(x)$ be any codeword of minimum weight $\wt(c(x)) = p(p-1)^2$. To prove that $c(x)$ is minimal, suppose there exists a non-zero codeword $c(y) \in \mathcal{C}_{D_{V_1}}$ such that $\operatorname{supp}(c(y)) \subseteq \operatorname{supp}(c(x))$. In terms of coordinate zero sets, this containment is equivalent to $Z(x) \subseteq Z(y)$, where $Z(x) = \{v \in D_{V_1} : c(x)_v = 0\} = D_{V_1} \cap H_x$. Since $|Z(x)| = p(p-1)$, $Z(y)$ cannot be empty, which implies that $\wt(c(y)) = p(p-1)^2$ and $|Z(y)| = p(p-1)$. If $y$ is not an $\mathbb{F}_p^*$-scalar multiple of $x$, the hyperplanes $H_x$ and $H_y$ are distinct, and their intersection $W = H_x \cap H_y$ is a $1$-dimensional subspace in $\mathbb{F}_p^3$. By Lemma~\ref{lem:subspace-intersection}, any $1$-dimensional subspace intersects $D_{V_1}$ in at most $p - 1$ points. Consequently:
\[
|Z(x) \cap Z(y)| = |D_{V_1} \cap (H_x \cap H_y)| \le p - 1.
\]
However, the inclusion $Z(x) \subseteq Z(y)$ forces $Z(x) = Z(x) \cap Z(y)$, which implies $|Z(x)| \le p - 1$. This directly contradicts $|Z(x)| = p(p-1) > p - 1$ for all primes $p \ge 3$. Therefore, $H_x = H_y$, implying $y = \alpha x$ for some $\alpha \in \mathbb{F}_p^*$. Hence, every codeword of minimum weight $p(p-1)^2$ is strictly minimal.

Finally, testing the Griesmer bound for dimension $k = 3$:
\[
\sum_{i=0}^2 \left\lceil \frac{d}{p^i} \right\rceil = \left\lceil \frac{p(p-1)^2}{1} \right\rceil + \left\lceil \frac{p(p-1)^2}{p} \right\rceil + \left\lceil \frac{p(p-1)^2}{p^2} \right\rceil
\]
\[
= p(p-1)^2 + (p-1)^2 + (p-1) = p^2(p-1) = n.
\]
Because the code length achieves the Griesmer bound with equality, $\mathcal{C}_{D_{V_1}}$ is Griesmer-optimal.
\end{proof}

\begin{Theorem}\label{thm:weight-hierarchy}
Let $D_{V_1} = uR \setminus u^2 R$ be the maximal independent defining set of $\Gamma(R)$, and let $\mathcal{C}_{D_{V_1}}$ be the associated $[p^2(p-1), 3, p(p-1)^2]_p$ linear code. The complete Wei generalized Hamming weight hierarchy $\{d_1, d_2, d_3\}$ of $\mathcal{C}_{D_{V_1}}$ is given explicitly by:
\[
d_1 = p(p-1)^2, \quad d_2 = p^3 - 2p^2 + 2p - 1 = (p-1)(p^2 - p + 1), \quad d_3 = p^2(p-1).
\]
Furthermore, the weights $d_1$ and $d_3$ attain the generalized Griesmer bound with equality, while $d_2$ achieves the bound with an optimal defect of strictly $1$:
\[
\sum_{i=0}^{3-r} \left\lceil \frac{d_r}{p^i} \right\rceil = 
\begin{cases}
n, & \text{for } r \in \{1, 3\}, \\
n + 1, & \text{for } r = 2.
\end{cases}
\]
\end{Theorem}

\begin{proof}
For an $[n, k, d]_p$ linear code $\mathcal{C}$ and an integer $1 \le r \le k$, the $r$-th generalized Hamming weight $d_r(\mathcal{C})$ is defined as: $d_r(\mathcal{C}) = \min \left\{ |\operatorname{supp}(\mathcal{H})| : \mathcal{H} \le \mathcal{C}, \; \dim(\mathcal{H}) = r \right\}$, where $\operatorname{supp}(\mathcal{H}) = \bigcup_{c \in \mathcal{H}} \operatorname{supp}(c)$.

By Theorem~\ref{thm:indep-code}, $\dim(\mathcal{C}_{D_{V_1}}) = 3$ and the minimum Hamming distance is $d_1 = d = p(p-1)^2$. For $r = k = 3$, the unique $3$-dimensional subcode is $\mathcal{C}_{D_{V_1}}$ itself. As shown in Case 1 of Theorem~\ref{thm:indep-code}, the evaluation vector $x = (0, 1, 0, 0)$ produces a codeword whose support is the entire coordinate set $D_{V_1}$. Thus $d_3(\mathcal{C}_{D_{V_1}}) = |D_{V_1}| = p^2(p-1)$.

It remains to determine $d_2(\mathcal{C}_{D_{V_1}})$. By the geometric characterization of generalized Hamming weights, for an evaluation code $\mathcal{C}_D$ generated by an ordered defining set $D \subset \mathbb{F}_p^m$, the $r$-th weight satisfies:
\[
d_r(\mathcal{C}_D) = |D| - \max \left\{ |D \cap H| : H \text{ is a subspace of } \mathbb{F}_p^m \text{ with } \operatorname{codim}(H) = r \right\}.
\]
For $r = 2$, this reduces to finding the maximum number of points in $D_{V_1}$ that simultaneously lie in two distinct linear hyperplanes $H_1, H_2 \subset \mathbb{F}_p^3$:
\[
d_2(\mathcal{C}_{D_{V_1}}) = |D_{V_1}| - \max_{H_1 \ne H_2} \left| D_{V_1} \cap (H_1 \cap H_2) \right|.
\]
The intersection $W = H_1 \cap H_2$ is a $1$-dimensional subspace in $\mathbb{F}_p^3$. By Lemma~\ref{lem:subspace-intersection}, the maximum intersection of $D_{V_1}$ with any $1$-dimensional subspace is strictly $M_2 = p - 1$. Subtracting this value from the code length yields:
\[
d_2(\mathcal{C}_{D_{V_1}}) = |D_{V_1}| - M_2 = p^2(p - 1) - (p - 1) = (p^2 - 1)(p - 1) = p^3 - 2p^2 + 2p - 1.
\]

Finally, we test the generalized Griesmer bound $n \ge \sum_{i=0}^{3-r} \lceil d_r / p^i \rceil$ for each weight $r \in \{1, 2, 3\}$:
\begin{itemize}
    \item For $r = 1$: $\sum_{i=0}^2 \lceil d_1 / p^i \rceil = p^2(p-1) = n$.
    \item For $r = 3$: $\sum_{i=0}^0 \lceil d_3 / p^i \rceil = d_3 = p^2(p-1) = n$.
    \item For $r = 2$: the Griesmer sum evaluates to:
    \[
    \sum_{i=0}^1 \left\lceil \frac{d_2}{p^i} \right\rceil = d_2 + \left\lceil \frac{d_2}{p} \right\rceil = (p^3 - 2p^2 + 2p - 1) + \left\lceil p^2 - 2p + 2 - \frac{1}{p} \right\rceil
    \]
    \[
    = (p^3 - 2p^2 + 2p - 1) + (p^2 - 2p + 2) = p^3 - p^2 + 1 = p^2(p-1) + 1 = n + 1.
    \]
\end{itemize}
Thus, the higher weight $d_2$ attains the generalized Griesmer bound with an optimal defect of strictly $1$, while $d_1$ and $d_3$ achieve the bound with equality.
\end{proof}

\begin{Remark}[Cryptographic Access Structures via Minimal Codewords]\label{rem:secret-sharing}
In Massey's secret-sharing framework \cite{Ashikhmin1998, Ding2007}, minimal authorized recovery coalitions correspond directly to minimal codewords possessing a non-zero coordinate at the secret position. For the two-weight evaluation code $\mathcal{C}_{D_{V_1}}$, while the $p - 1$ maximum-weight codewords have support spanning the entire coordinate set $D_{V_1}$, each of the $(p-1)(p^2 + p)$ minimum-weight codewords is strictly minimal by Theorem~\ref{thm:indep-code}. Consequently, each minimum-weight codeword defines a minimal authorized coalition whose participants can reconstruct the secret, with no proper sub-coalition possessing enough information to do so.
\end{Remark}

\begin{Example}\label{ex:defining-codes}
For $p = 3$, consider the ring $R = \mathbb{F}_3 + u\mathbb{F}_3 + u^2\mathbb{F}_3 + u^3\mathbb{F}_3$ ($u^4=0$). 
\begin{enumerate}
    \item Maximal Clique Code $\mathcal{C}_{D_K}$: The defining set $D_K = u^2 R \setminus \{0\}$ contains $n = 3^2 - 1 = 8$ vertices:
    \[
    D_K = \{u^2, \, 2u^2, \, u^3, \, 2u^3, \, u^2+u^3, \, 2u^2+u^3, \, u^2+2u^3, \, 2u^2+2u^3\}.
    \]
    Evaluating the linear forms $c(x) = (x \cdot [d])_{d \in D_K}$ for all evaluation vectors $x \in \mathbb{F}_3^4$ yields a ternary linear code $\mathcal{C}_{D_K}$ with parameters $[8, 2, 6]_3$. Its complete weight enumerator is $W_{\mathcal{C}_{D_K}}(z) = 1 + 8z^6$. Since $\frac{w_{\max}}{w_{\min}} = \frac{6}{6} = 1 < \frac{3}{2}$, $\mathcal{C}_{D_K}$ is a minimal linear code meeting the Griesmer bound:
    \[
    \left\lceil \frac{6}{1} \right\rceil + \left\lceil \frac{6}{3} \right\rceil = 6 + 2 = 8.
    \]

    \item Maximal Independent Set Code $\mathcal{C}_{D_{V_1}}$: The defining set $D_{V_1} = uR \setminus u^2 R$ has length $n = 3^2(3-1) = 18$. Evaluating the linear forms $c(x) = (x \cdot [d])_{d \in D_{V_1}}$ for all $x \in \mathbb{F}_3^4$ yields an optimal two-weight $[18, 3, 12]_3$ linear code with weights $w_{\min} = 3(3-1)^2 = 12$ and $w_{\max} = 3^2(3-1) = 18$. Its complete weight enumerator is:
    \[
    W_{\mathcal{C}_{D_{V_1}}}(z) = 1 + 24 z^{12} + 2 z^{18}.
    \]
    By Theorem~\ref{thm:weight-hierarchy}, the generalized Hamming weight hierarchy is $\{d_1, d_2, d_3\} = \{12, 14, 18\}$. Testing the generalized Griesmer sums:
    \begin{itemize}
        \item For $r = 1$: $\lceil 12/1 \rceil + \lceil 12/3 \rceil + \lceil 12/9 \rceil = 12 + 4 + 2 = 18 = n$.
        \item For $r = 2$: $\lceil 14/1 \rceil + \lceil 14/3 \rceil = 14 + 5 = 19 = n + 1$ (defect $1$).
        \item For $r = 3$: $\lceil 18/1 \rceil = 18 = n$.
    \end{itemize}
\end{enumerate}
\end{Example}

\section*{Conclusion }
In this paper, we investigated the algebraic coding-theoretic properties of the zero-divisor graph $\Gamma(R)$ over the finite local ring $R = \mathbb{F}_{p}[x]/\langle x^{4}\rangle \cong \mathbb{F}_{p}+u\mathbb{F}_{p}+u^{2}\mathbb{F}_{p}+u^{3}\mathbb{F}_{p}$ ($u^4=0$). Utilizing an equitable partition of $Z^*(R)$, we fully parameterized binary incidence and cycle codes, enhanced dual error detection via induced bipartite subgraphs, and evaluated $p$-ary codes generated by the adjacency and Laplacian matrices, including closed-form weight enumerators, LCD properties, and full automorphism groups. Furthermore, using extremal subgraphs as defining sets, we constructed two families of linear codes achieving the Griesmer bound with equality, consisting of a constant-weight minimal code and an optimal two-weight code. Future work will focus on generalizing these combinatorial constructions to finite local rings with arbitrary nilpotency index $\mathbb{F}_{p}[x]/\langle x^{n}\rangle$ ($n \ge 5$). Additionally, it would be interesting to explore the cryptographic applications of the resulting minimal codes in designing secret-sharing schemes with democratic access structures.

\section*{Declarations}
The authors received no external funding or financial support for conducting this research or preparing this manuscript.  The graph visualizations illustrating the equitable partition and vertex-level topology are provided in Appendix~A. All authors contributed equally to the theoretical formulation, derivations, and drafting of the manuscript, and all approved the final version for submission.

\section*{Conflict of Interest}
The authors declare that they have no conflicts of interest, financial or non-financial, directly or indirectly related to the research presented in this paper. No personal, commercial, or academic affiliations have influenced the objectivity or integrity of the reported work.

\newpage 
\appendix

\section{Graph Visualizations of $\Gamma(R)$}
\label{app:figures}

This appendix presents graphical visualizations of the zero-divisor graph $\Gamma(R)$ for the local ring $R = \mathbb{F}_p + u\mathbb{F}_p + u^2\mathbb{F}_p + u^3\mathbb{F}_p$ ($u^4=0$). Figure 1 depicts the equitable partition of the non-zero zero-divisors $Z^*(R)$ into its seven $u$-adic classes, emphasizing the central dominating role of the subring ideal class $S_{u^3}$. Figure~2 illustrates the complete vertex-level topology of $\Gamma(R)$ for the case $p=3$, displaying all $26$ vertices and $64$ incident edges.
\begin{figure}[htbp]\label{fig:appendix_fig1}
\centering
\begin{tikzpicture}[
    thick,
    centerNode/.style={
        ellipse, minimum width=1.8cm, minimum height=2.6cm,
        draw=orange!80!black, top color=orange!10, bottom color=orange!30,
        font=\bfseries, align=center
    },
    outerNode/.style={
        ellipse, minimum width=2.0cm, minimum height=1.4cm,
        draw=teal!80!black, top color=teal!5, bottom color=teal!20,
        font=\small, align=center
    },
    divider/.style={thick, draw=gray!70, dashed},
    connection/.style={very thick, draw=gray!60},
    specialConnection/.style={very thick, draw=gray!70}
]
    \def\R{4.5}
    \node[centerNode] (Center) at (0,0) {$S_{u^3}$};
    \draw[divider,solid] ($(Center.north)+(0,-0.2)$) -- ($(Center.south)+(0,0.2)$);

    \foreach \ang/\name/\lab/\hasline in {
        0/N1/{$S_{u^2+u^3}$}/1,
        60/N2/{$S_{u^2}$}/1,
        120/N3/{$S_{u}$}/0,
        180/N4/{$S_{u+u^2}$}/0,
        240/N5/{$S_{u+u^3}$}/0,
        300/N6/{$S_{u+u^2+u^3}$}/0
    } {
        \node[outerNode] (\name) at (\ang:\R) {\lab};
        \draw[connection] (Center) -- (\name);
        \ifnum\hasline=1
            \draw[divider,solid] ($(\name.north)+(0,-0.1)$) -- ($(\name.south)+(0,0.1)$);
        \fi
    }
    \draw[specialConnection] (N1) -- (N2);
\end{tikzpicture}
\caption{Zero-divisor graph of $R = \mathbb{F}_{p} + u\mathbb{F}_{p} + u^{2}\mathbb{F}_{p} + u^{3}\mathbb{F}_{p}$.}
\label{Fig1}
\end{figure}
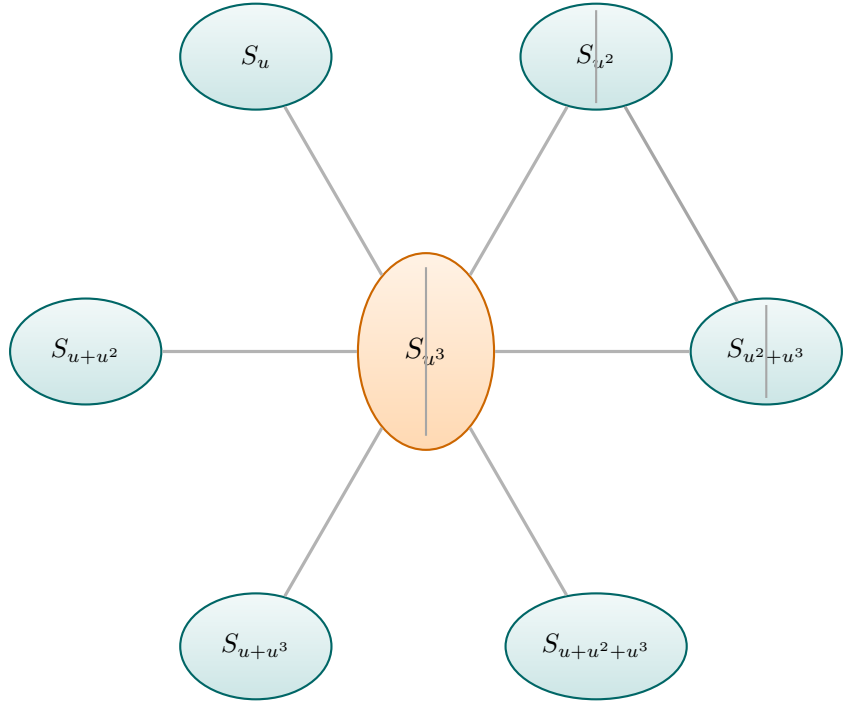
\clearpage

\begin{figure}[htbp]\label{fig:appendix_fig2}
\centering
\begin{tikzpicture}[scale=0.9, transform shape]
\draw[thick] (0,0) ellipse (1.4 and 2.4);
\node[above=2.6cm] at (0,0) {$S_{u^3}$};

\foreach \angle/\name in {
    0/S_{u^2+u^3}, 60/S_{u^2}, 120/S_{u},
    180/S_{u+u^2}, 240/S_{u+u^3}, 300/S_{u+u^2+u^3}
}{
    \draw[thick] ({6*cos(\angle)},{6*sin(\angle)}) ellipse (1.2 and 2.2);
    \node[above=2.4cm] at ({6*cos(\angle)},{6*sin(\angle)}) {$\name$};
}

\node[circle, fill=black, inner sep=1.5pt, label=above:$u^3$] (a0u) at (0,1) {};
\node[circle, fill=black, inner sep=1.5pt, label=below:$2u^3$] (a0v) at (0,-1) {};

\node[circle, fill=black, inner sep=1.5pt, label=right:{$u^2+u^3$}] (a11) at (6,1.2) {};
\node[circle, fill=black, inner sep=1.5pt, label=right:{$2u^2+u^3$}] (a12) at (6,0.4) {};
\node[circle, fill=black, inner sep=1.5pt, label=right:{$u^2+2u^3$}] (a13) at (6,-0.4) {};
\node[circle, fill=black, inner sep=1.5pt, label=right:{$2u^2+2u^3$}] (a14) at (6,-1.2) {};

\node[circle, fill=black, inner sep=1.5pt, label=right:{$u^2$}]  (a21) at ({6*cos(60)},{6*sin(60)+0.6}) {};
\node[circle, fill=black, inner sep=1.5pt, label=right:{$2u^2$}] (a22) at ({6*cos(60)},{6*sin(60)-0.6}) {};

\node[circle, fill=black, inner sep=1.5pt, label=left:{$u$}]  (a31) at ({6*cos(120)},{6*sin(120)+0.6}) {};
\node[circle, fill=black, inner sep=1.5pt, label=left:{$2u$}] (a32) at ({6*cos(120)},{6*sin(120)-0.6}) {};

\node[circle, fill=black, inner sep=1.5pt, label=left:{$u+u^2$}]   (a41) at ({6*cos(180)},{6*sin(180)+1.2}) {};
\node[circle, fill=black, inner sep=1.5pt, label=left:{$2u+u^2$}]  (a42) at ({6*cos(180)},{6*sin(180)+0.4}) {};
\node[circle, fill=black, inner sep=1.5pt, label=left:{$u+2u^2$}]   (a43) at ({6*cos(180)},{6*sin(180)-0.4}) {};
\node[circle, fill=black, inner sep=1.5pt, label=left:{$2u+2u^2$}] (a44) at ({6*cos(180)},{6*sin(180)-1.2}) {};

\node[circle, fill=black, inner sep=1.5pt, label=left:{$u+u^3$}]   (a51) at ({6*cos(240)},{6*sin(240)+1.2}) {};
\node[circle, fill=black, inner sep=1.5pt, label=left:{$2u+u^3$}]  (a52) at ({6*cos(240)},{6*sin(240)+0.4}) {};
\node[circle, fill=black, inner sep=1.5pt, label=left:{$u+2u^3$}]   (a53) at ({6*cos(240)},{6*sin(240)-0.4}) {};
\node[circle, fill=black, inner sep=1.5pt, label=left:{$2u+2u^3$}] (a54) at ({6*cos(240)},{6*sin(240)-1.2}) {};

\foreach \i/\y/\lbl in {
  1/1.4/$u+u^2+u^3$, 2/1.0/$2u+u^2+u^3$, 3/0.6/$u+2u^2+u^3$, 4/0.2/$u+u^2+2u^3$,
  5/-0.2/$2u+2u^2+u^3$, 6/-0.6/$2u+u^2+2u^3$, 7/-1.0/$u+2u^2+2u^3$, 8/-1.4/$2u+2u^2+2u^3$
}{
  \node[circle, fill=black, inner sep=1.3pt, label=right:{\lbl}]
    (a6\i) at ({6*cos(300)},{6*sin(300)+\y}) {};
}

\foreach \c in {a0u,a0v}{
  \foreach \n in {a11,a12,a13,a14}{\draw (\c) -- (\n);}
  \foreach \n in {a21,a22}{\draw (\c) -- (\n);}
  \foreach \n in {a31,a32}{\draw (\c) -- (\n);}
  \foreach \n in {a41,a42,a43,a44}{\draw (\c) -- (\n);}
  \foreach \n in {a51,a52,a53,a54}{\draw (\c) -- (\n);}
  \foreach \i in {1,...,8}{\draw (\c) -- (a6\i);}
}

\foreach \nA in {a11,a12,a13,a14}{
  \foreach \nB in {a21,a22}{\draw (\nA) -- (\nB);}
}

\draw (a0u) -- (a0v);
\draw[red, bend left=20, line width=0.8pt] (a11) to (a12);
\draw[blue, bend left=30, line width=0.8pt] (a11) to (a13);
\draw[green!70!black, bend left=40, line width=0.8pt] (a11) to (a14);
\draw[orange, bend left=20, line width=0.8pt] (a12) to (a13);
\draw[purple, bend left=30, line width=0.8pt] (a12) to (a14);
\draw[teal, bend left=20, line width=0.8pt] (a13) to (a14);
\draw (a21) -- (a22);

\end{tikzpicture}
\caption{Zero-divisor graph of $R = \mathbb{F}_3 + u\mathbb{F}_3 + u^2\mathbb{F}_3 + u^3\mathbb{F}_3$.}
\label{Fig2}
\end{figure}
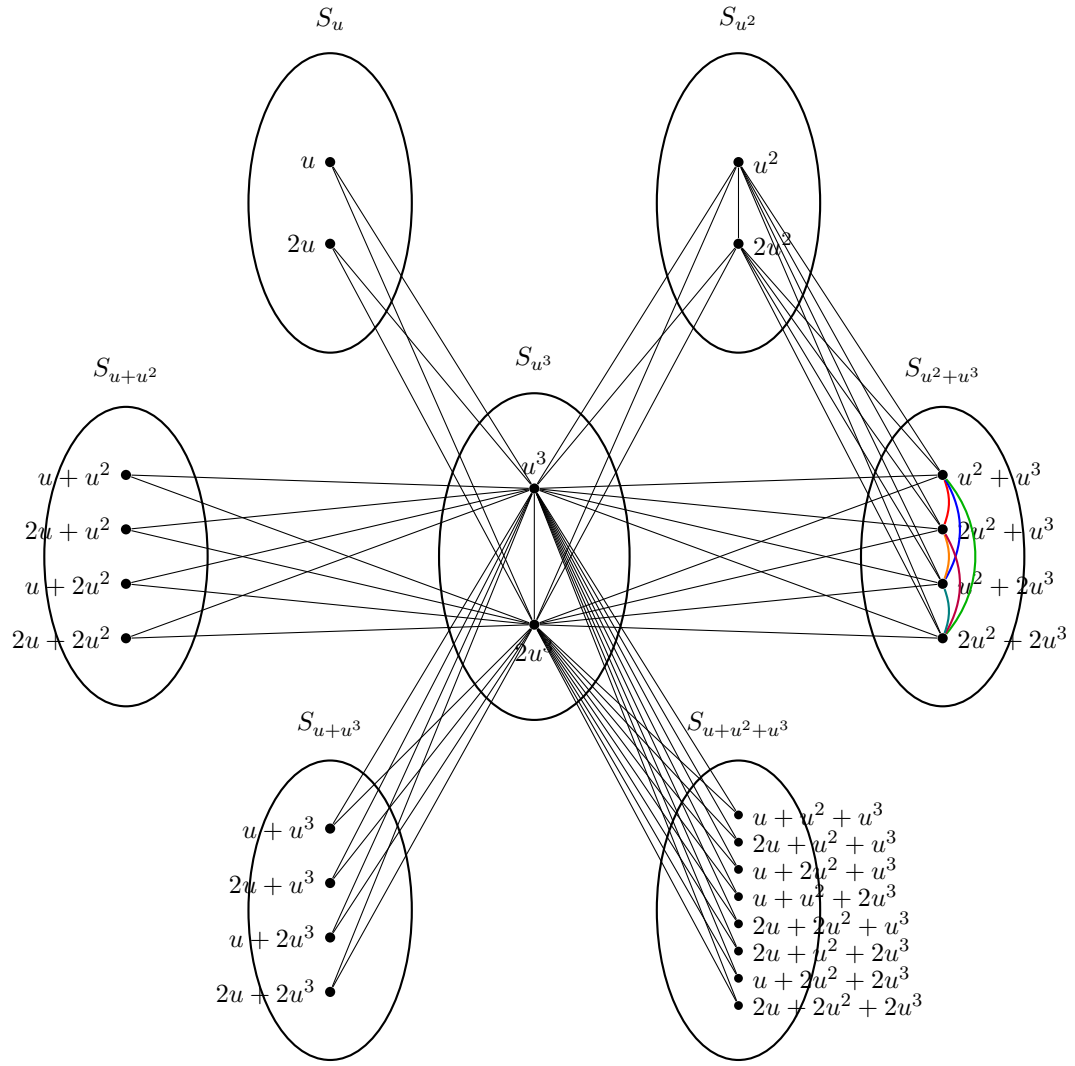

\end{document}